\documentclass[10pt,a4paper]{article}
\usepackage[utf8]{inputenc}
\usepackage[english]{babel}
\usepackage{amsmath}
\usepackage{amsfonts}
\usepackage{amssymb}
\usepackage{amsthm}
\usepackage{graphicx}
\usepackage{fullpage}
\usepackage{eucal}
\usepackage{float}

\usepackage{mathrsfs} % na poradne H
\usepackage{mathtools}
\usepackage{color}
\usepackage{stackrel}
\usepackage{tensor}

\graphicspath{{figs/}}

\newcommand{\deff}{\vcentcolon=}  % hezke zarovnani :=
\newcommand{\deffin}{=\vcentcolon}
\newcommand{\N}{\mathbb{N}}   % ciselna telesa

\newcommand{\R}{\mathbb{R}}

\newcommand{\der}{\mathrm{d}}
\newcommand{\dd}{\,\mathrm{d}} % diferencial, i v integralech
\newcommand{\ie}{\emph{i.e.}}
\newcommand{\eg}{\emph{e.g.}}
\newcommand{\cf}{\emph{cf.}}

\newcommand{\supp}{\mathop{\mathrm{supp}}\nolimits}
\newcommand{\Dom}{\mathrm{Dom}}
\newcommand{\Hilbert}{\mathcal{H}}
\newcommand{\dist}{\mathop{\mathrm{dist}}\nolimits}
\newcommand{\esssup}{\mathop{\mathrm{ess\;\!sup}}}
\newcommand{\essinf}{\mathop{\mathrm{ess\;\!inf}}}
\newcommand{\espec}{\sigma_{\mathrm{ess}}}%

\newcommand{\eps}{\varepsilon}
\newcommand{\sii}{L^2}

\newcommand{\dvtheta}{\Theta '}
\newcommand{\Om}{\Omega}

\newcommand{\vertiii}[1]{{\left\vert\kern-0.25ex\left\vert\kern-0.25ex\left\vert #1 
    \right\vert\kern-0.25ex\right\vert\kern-0.25ex\right\vert}}
\newtheorem{theorem}{Theorem}[section]
\newtheorem{corollary}[theorem]{Corollary}
\newtheorem{lemma}[theorem]{Lemma}
\newtheorem{proposition}[theorem]{Proposition}

\theoremstyle{definition}

\newtheorem{example}{Example}

\newtheorem{Assumption}{Assumption}
\newtheorem{remark}[theorem]{Remark}
\newtheorem{open}[theorem]{Open Problem}
\numberwithin{equation}{section}

\newcommand{\Hm}[1]{\leavevmode{\marginpar{\tiny%
$\hbox to 0mm{\hspace*{-0.5mm}$\leftarrow$\hss}%
\vcenter{\vrule depth 0.1mm height 0.1mm width \the\marginparwidth}%
\hbox to
0mm{\hss$\rightarrow$\hspace*{-0.5mm}}$\\\relax\raggedright #1}}}

\begin{document}
%
%-------%
% TITLE %
%-------%
%------------------------------------------%
%------------------------------------------%
\title{\textbf{\LARGE Quantum strips in higher dimensions}}
\author{David Krej\v{c}i\v{r}{\'\i}k\,$^a$
\ and \ Kate\v{r}ina Zahradov\'a\,$^b$}
\date{\small 
\begin{quote}
\emph{
\begin{itemize}
\item[$a)$] 
Department of Mathematics, Faculty of Nuclear Sciences and 
Physical Engineering, Czech Technical University in Prague, 
Trojanova 13, 12000 Prague 2, Czechia;
david.krejcirik@fjfi.cvut.cz.%
\item[$b)$] 
Department of Theoretical Physics, Nuclear Physics Institute, 
Czech Academy of Sciences, 25068 \v{R}e\v{z}, Czechia, \& School of Mathematical Sciences, Queen Mary University of London, London E1 4NS, United Kingdom
k.zahradova@qmul.ac.uk.%
\end{itemize}
}
\end{quote}
31 May 2019}
\maketitle
\begin{abstract}
\noindent
We consider the Dirichlet Laplacian 
in unbounded strips on ruled surfaces in any space dimension.
We locate the essential spectrum under the condition that the strip is asymptotically flat.
If the Gauss curvature of the strip equals zero, 
we establish the existence of discrete spectrum 
under the condition that the curve along which the strip is built is not a geodesic.
On the other hand, if it is a geodesic 
and the Gauss curvature is not identically equal to zero,
we prove the existence of Hardy-type inequalities.
%provided that the strip is not too thick.
We also derive an effective operator for thin strips,
which enables one to replace the spectral problem
for the Laplace-Beltrami operator on the two-dimensional surface 
by a one-dimensional Schr\"odinger operator whose potential 
is expressed in terms of curvatures.

In the appendix, we establish a purely geometric fact
about the existence of relatively parallel adapted frames for any curve
under minimal regularity hypotheses. 
%
%\bigskip
%\begin{itemize}
%\item[\textbf{Keywords:}]
%\item[\textbf{MSC (2010):}]
%\end{itemize}
%
\end{abstract}
%
%------------------------------------------%
%------------------------------------------%
 
%---------------------%
\section{Introduction}
%---------------------%
%
The interplay between the geometry of a Euclidean domain  
or a Riemannian manifold and spectral properties of
underlying differential operators constitute one of 
the most fascinating problems in mathematical sciences
over the last centuries.
A special allure is without doubts due to the emotional 
impacts the shape of objects has over a person's perception of the world,
while the spectrum typically admits direct physical interpretations.
With the advent of nanoscience, new layouts like unbounded tubes
have become highly attractive in the context of guided quantum particles
and brought unprecedented spectral-geometric phenomena.

Let us demonstrate the attractiveness of the subject
on the simplest non-trivial 
model of two-dimensional waveguides that we nicknamed
\emph{quantum strips on surfaces} in~\cite{K1}.

\begin{itemize}
\item
The spectrum of the Laplacian in
a \emph{straight} strip $\Omega_0 \deff \R \times (-a,a)$ of half-width $a>0$,
subject to uniform boundary conditions,  
was certainly known to Helmholtz if not already to Laplace.
For Dirichlet boundary conditions, 
the spectrum coincides with the semi-axis $[E_1,\infty)$,
where the spectral threshold $E_1 \deff (\frac{\pi}{2a})^2$ is positive,
indicating thus an interpretation in terms of a semiconductor.    
\item
In 1989, Exner and \v{S}eba~\cite{ES} demonstrated
that \emph{bending} the strip locally in the plane 
does not change the essential spectrum
but generates discrete eigenvalues below~$E_1$.  
In other words, realising the bent strip as a tubular neighbourhood
of radius~$a$ of an unbounded curve in the plane, 
the curvature of the curve induces a sort of \emph{attractive} interaction,
which diminishes the spectral threshold 
and leads to quantum bound states
(without classical counterparts).
We refer to~\cite{KKriz} for a survey on the bent strips.
\item
What is the effect of the curvature of the ambient space on the spectrum?
More specifically, embedding the strip in a two-dimensional Riemannian manifold
instead of~$\R^2$, %(of zero curvature), 
how does the spectrum of the Dirichlet Laplacian change?
In 2003, it was demonstrated by one of the present authors~\cite{K1} 
that \emph{positive} curvature of the ambient manifold
still acts as an \emph{attractive} interaction,
even if the (geodesic) curvature of the underlying curve is zero.
\item
On the other hand, in 2006, the same author~\cite{K3}  
showed that the effect
of \emph{negative} ambient curvature is quite opposite
in the sense that it now acts as a sort of \emph{repulsive} interaction.
More specifically, if the Gauss curvature vanishes at infinity 
and the underlying curve is a geodesic, 
the spectrum is $[E_1,\infty)$ like in the straight strip, 
but the Dirichlet Laplacian additionally satisfies Hardy-type inequalities.  
\item
As a matter of fact, the presence of Hardy-type inequalities
was proved in~\cite{K3} only for strips on \emph{ruled} surfaces,
but the robustness of the result for general negatively curved surfaces
was further confirmed in~\cite{KK3}.
Now, a strip on a ruled surface can be alternatively realised 
as a \emph{twisted} (and possibly also bent) strip in~$\R^3$,
so the repulsiveness effect is analogous to 
the presence of Hardy-type inequalities in solid waveguides~\cite{EKK}
(see also \cite{K6-with-erratum}).   
\end{itemize}

The primary objective of this paper is to extend the results for strips
on ruled surfaces to higher dimensions, meaning that the twisted
and bent two-dimensional strip is embedded in~$\R^d$ with any $d \geq 3$. 
A secondary goal is to improve and unify the known results 
even in dimension $d=3$ by considering more general underlying curves.
More specifically, we consider strips built with help
of a \emph{relatively parallel adapted frame} (which always exists) 
instead of the customary Frenet frame (which does not need to exist).
Since this purely geometric construction, which we have decided
to present in the appendix, does not seem to be well known
(definitely not for higher-dimensional curves),
we believe that the material will be of independent interest 
(not only) for the quantum-waveguide community.

\begin{figure}[H]
\begin{center}
\bigskip
\includegraphics[width=0.9\textwidth]{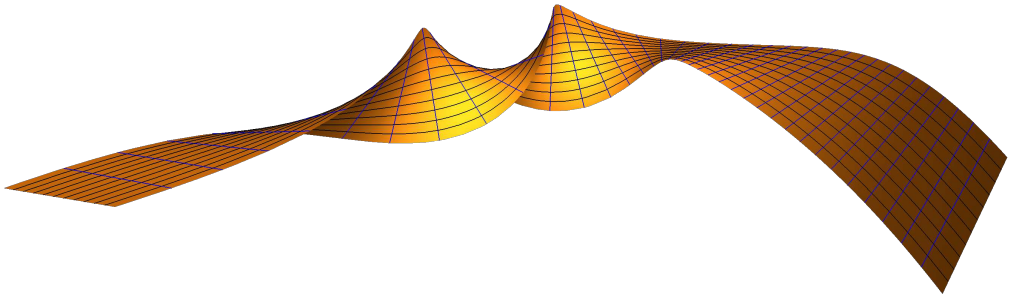}
\caption{\small A simultaneously bent and twisted strip.}
\label{Fig.both}
\end{center}
\end{figure}

The structure of the paper is as follows.
In Section~\ref{Sec.strips} we introduce the Dirichlet Laplacian
in strips on ruled surfaces in any space dimension 
under minimal regularity hypotheses.	
The essential spectrum of asymptotically flat strips
is located in Section~\ref{Sec.ess}.
The effects of bending and twisting are investigated
in Sections~\ref{Sec.bent} and~\ref{Sec.twist}, respectively.
In Section~\ref{Sec.thin} we show that,
in the limit when the width of the strip tends to zero,
the Dirichlet Laplacian converges in a norm resolvent sense
to a one-dimensional Schr\"odinger operator
whose potential contains information about 
the deformations of twisting and bending.
Appendix~\ref{Sec.app} is devoted to the construction
of a relatively parallel adapted frame 
for an arbitrary curve.

%--------------------------------------------------------%
\section{Definition of quantum strips}\label{Sec.strips}
%--------------------------------------------------------%
%

\subsection{The reference curve}
Given any positive integer~$n$,
let $\Gamma:\R\to\R^{n+1}$ be a curve of class $C^{1,1}$
which is (without loss of generality) parameterised by its arc-length
(\ie\ $|\Gamma'(s) |=1$ for all $s \in \R$).
By the regularity hypothesis, the tangent vector field $T\deff\Gamma'$
is differentiable almost everywhere.
Moreover, in Appendix~\ref{Sec.app}, 
we show that there exist~$n$ almost-everywhere differentiable 
normal vector fields $N_1,\dots,N_n$ such that
\begin{equation}\label{frame}
\begin{pmatrix}
    T \\
    N_1 \\
    \vdots \\
    N_n
\end{pmatrix}'
=
\begin{pmatrix}
    0 & k_1 & \dots  & k_n \\
   -k_1 &  0 & \dots  & 0 \\
    \vdots & \vdots & \ddots & \vdots \\
    -k_n  & 0 & \dots  & 0
\end{pmatrix}
\begin{pmatrix}
T \\
    N_1 \\
    \vdots \\
    N_n
\end{pmatrix}
,
\end{equation}
where $k_1,\dots,k_n:\R \to \R$ are locally bounded functions.
Introducing the $n$-tuple $k\deff(k_1,\dots,k_n)$ 
and calling it the \emph{curvature vector},
we have $k_1^2 + \dots + k_n^2 = \kappa^2$
with $\kappa \deff |\Gamma''|$ being the \emph{curvature} of~$\Gamma$.

Since the derivative~$N_j'$ is tangential for every $j\in\{1,\dots,n\}$,
the normal vectors rotate along the curve~$\Gamma$ 
only whatever amount is necessary to remain normal.
In fact, each normal vector~$N_j$ is translated 
along~$\Gamma$ as close to a parallel transport as possible
without losing normality.
For this reason, 
and in analogy with the three-dimensional setting~\cite{Bishop_1975},
each vector field~$N_j$ is called \emph{relatively parallel}
and the $(n+1)$-tuple $(T,N_1,\dots,N_n)$
is called a \emph{relatively parallel adapted frame}.
Notice that contrary to the standard Frenet frame
which requires a higher regularity~$C^{n+1}$ 
and the non-degeneracy condition $\kappa>0$,
the relatively parallel adapted frame always exists
under the minimal hypothesis~$C^{1,1}$.

\subsection{The strip as a surface in the Euclidean space}
Recall the definition $\Omega_0 \deff \R \times (-a,a)$ for a straight strip.
Isometrically embedding~$\Omega_0$ to $\R^{n+1}$,
we can think of~$\Omega_0$ as a surface in $\R^{n+1}$
obtained by parallelly translating the segment $(-a,a)$
along a straight line.
We define a general \emph{curved strip}~$\Omega$ in $\R^{n+1}$ 
as the ruled surface obtained
by translating the segment $(-a,a)$ along~$\Gamma$ 
with respect to a generic normal field 
\begin{equation}\label{generic}
  N_\Theta \deff \Theta_1 N_1 + \dots + \Theta_n N_n
  \,,
\end{equation}
where $\Theta_j : \R \to \R$ with $j\in\{1,\dots,n\}$
are such scalar functions that $\Theta_j \in C^{0,1}(\R)$ and
\begin{equation}\label{Eq sum of theta 2}
  \Theta_1^2 + \dots + \Theta_n^2 = 1
  \,.
\end{equation}
More specifically, we set 
\begin{equation}\label{strip}
  \Omega \deff \big\{
  \Gamma (s) + N_\Theta(s) \, t : \ (s,t) \in \Omega_0
  \big\}
  \,.
\end{equation}
In this way, $\Omega$~can be clearly understood 
as a deformation of the straight strip~$\Omega_0$,
see Figure~\ref{Fig.both}. 
 
We construct the $n$-tuple $\Theta \deff (\Theta_1, \dots, \Theta_n)$
and call it the \emph{twisting vector}. 
We naturally write $|\Theta'|\deff(\Theta_1'^2+\dots+\Theta_n'^2)^{1/2}$.
If the twisting vector~$\Theta$ is constant, \ie\ $\Theta'=0$, 
so that the vector field~$N_\Theta$ is relatively parallel,
we say that the strip~$\Omega$ is \emph{untwisted} or \emph{purely bent} 
(including the trivial situation $\kappa=0$
when~$\Omega$ can be identified with the straight strip~$\Omega_0$).
See Figure~\ref{Fig.bend} for a purely bent planar strip
and Figure~\ref{Fig.helix} (middle) for a purely bent non-planar strip.

\begin{figure}[H]
\begin{center}
\bigskip
\includegraphics[width=0.9\textwidth]{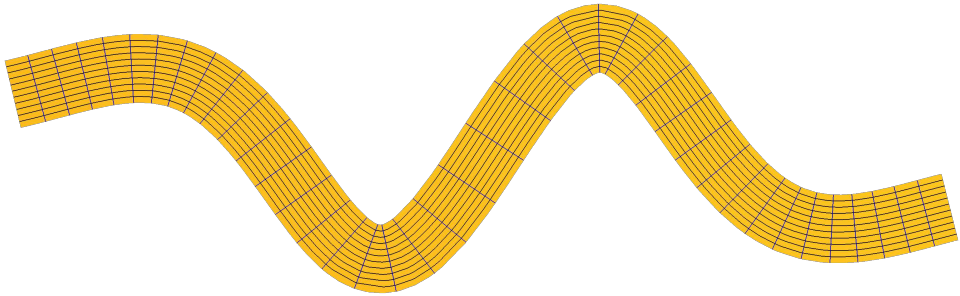}
\caption{\small A purely bent (planar) strip.}
\label{Fig.bend}
\end{center}
\end{figure}

On the other hand, if the scalar product of 
the curvature and twisting vectors vanishes,
\ie\ $k \cdot \Theta \deff k_1 \Theta_1 + \dots + k_n \Theta_n = 0$,
we say that the strip is \emph{unbent} or \emph{purely twisted} 
(including again the trivial situation $\kappa=0$ and $\Theta'=0$
when~$\Omega$ can be identified with the straight strip~$\Omega_0$).
See Figure~\ref{Fig.twist} for a purely twisted strip along a straight line 
and Figure~\ref{Fig.helix} (right) for a purely twisted strip along a space curve.

\begin{figure}[H]
\begin{center}
\bigskip
\includegraphics[width=0.9\textwidth]{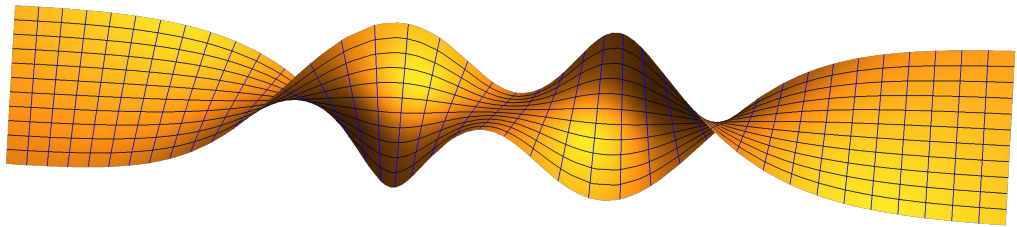}
\caption{\small A purely twisted strip.}
\label{Fig.twist}
\end{center}
\end{figure}

Notice that unbent and untwisted does not necessarily mean
that~$\Omega$ and~$\Omega_0$ are isometric
(think of a planar non-straight curve~$\Gamma$ in~$\R^3$
and choose for~$N_\Theta$ the binormal vector field),
see Figure~\ref{Fig.unboth}.

\begin{figure}[H]
\begin{center}
\bigskip
\includegraphics[width=0.9\textwidth]{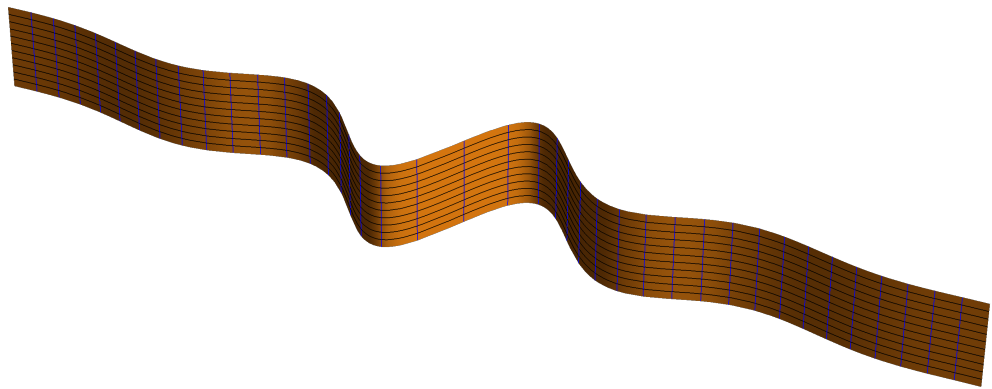}
\caption{\small An unbent untwisted strip.}
\label{Fig.unboth}
\end{center}
\end{figure}

Finally, Figure~\ref{Fig.helix} provides an example of 
a (non-planar) bent strip,
which is twisted or untwisted according to whether~$N_\Theta$
is relatively parallel or not, respectively.

\begin{remark}
Let us provide geometrical interpretations
to the crucial quantities $k\cdot\Theta$ and $\Theta'$
and supporting in this way the terminology introduced above.
Interpreting~$\Gamma$ as a curve on the surface~$\Omega$,
it is easily seen that $k\cdot\Theta$ is just 
the \emph{geodesic curvature} of~$\Gamma$.
Hence, the strip is unbent if, and only if, $\Gamma$ is a geodesic on~$\Omega$.
At the same time, 
the \emph{Gauss curvature} of the surface~$\Omega$
equals $-|\Theta'|^2/f^4$, where~$f$ is given in~\eqref{Eq metric} below. 
In accordance with a general result 
for ruled surfaces (\cf~\cite[Prop.~3.7.5]{Kli}),
we observe that this intrinsic curvature 
of the ambient manifold~$\Omega$ is always non-positive.
%as expected for a ruled surface (\cf~\cite[Prop.~3.7.5]{Kli}).
Moreover, the strip is untwisted if, and only if, 
the surface~$\Omega$ is flat in the sense that the Gauss curvature
is identically equal to zero.
\end{remark}
\begin{figure}[H]
%\begin{center}
\rule{-5ex}{0ex}
\begin{tabular}{ccc}%
\includegraphics[width=0.33\textwidth]{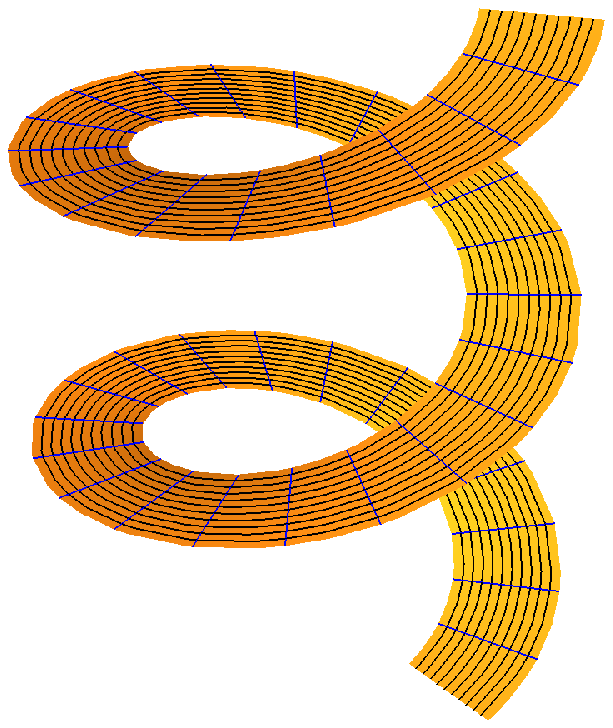}%
&\includegraphics[width=0.33\textwidth]{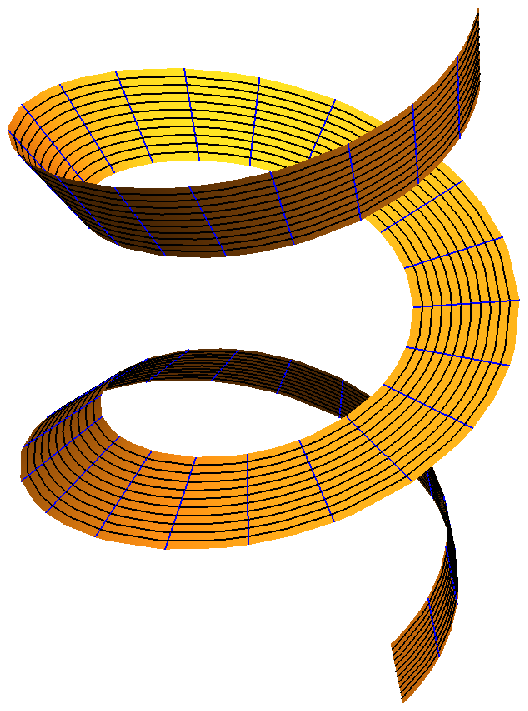}%
&\includegraphics[width=0.33\textwidth]{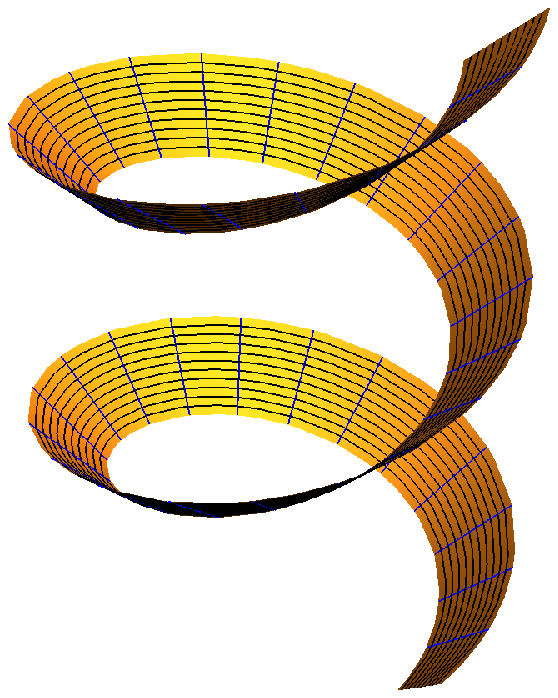}%
\\
$k\cdot\Theta\not=0 \ \land \ \Theta'\not=0$
& $k\cdot\Theta\not=0 \ \land \ \Theta'=0$
& $k\cdot\Theta=0 \ \land \ \Theta'\not=0$
\end{tabular}%
\caption{\small Strips built along a helix. 
Left: simultaneously bent and twisted version 
(the principal normal of the Frenet frame is used).
Middle: bent and untwisted version 
(a relatively parallel frame is used).
Right: unbent and twisted version
(a sum of the principal normal and binormal of the Frenet frame is used).}%
\label{Fig.helix}%
%\end{center}
%
\end{figure}

\subsection{The strip as a Riemannian manifold}
Further conditions must be imposed on the geometry of~$\Omega$
in order to identify the curved strip with a Riemannian manifold. 
To this aim, let us introduce the mapping 
$\mathscr{L}:\R^2 \to \R^{n+1}$ defined by
(\cf~\eqref{strip})
\begin{equation}\label{Eq curved strip}
  \mathscr{L}(s,t)
  \deff \Gamma (s) + N_\Theta(s) \, t  
  \,,
\end{equation}
so that $\Omega = \mathscr{L}(\Omega_0)$.
The blue segments in the figures represent 
the geodesics $t \mapsto \mathscr{L}(s,t)$
for various choices of~$s$,
while the black lines correspond to 
the curves $s \mapsto \mathscr{L}(s,t)$
parallel to~$\Gamma$ at distance~$|t|$. 

Consider the metric 
$g \deff \nabla\mathscr{L} \cdot (\nabla\mathscr{L})^T$,
where the dot denotes the scalar product in~$\R^{n+1}$.
A simple computation using~\eqref{frame} yields
\begin{equation}\label{Eq metric}
g = \begin{pmatrix}
f^2 & 0 \\
0 & 1
\end{pmatrix}
\qquad\mbox{with}\qquad
  f(s,t) \deff \sqrt{
  \big[ 1 - t \, k(s) \cdot \Theta(s) \big]^2 
  + t^2  \, |\Theta'(s)|^2
  }
  \,.
\end{equation}
Let us now strengthen our standing hypotheses.

\begin{center}
\fbox{%
\begin{minipage}{0.98\textwidth}
\begin{Assumption}\label{Ass.Ass}
Let $\Gamma \in C^{1,1}(\R;\R^{n+1})$ and $\Theta \in C^{0,1}(\R;\R^{n})$.
Suppose~\eqref{Eq sum of theta 2}, $k\cdot\Theta \in L^\infty(\R)$ and 
\begin{equation*}%\label{Ass}
  a \, \| k\cdot\Theta \|_{L^\infty(\R)} < 1  
  \,.
\end{equation*}
\end{Assumption}
\end{minipage}%
}
\end{center}

It follows from Assumption~\ref{Ass.Ass}
that the Jacobian~$f$ never vanishes, namely
\begin{equation}\label{Ass.bound}
  f(s,t) \geq 1-a \, \| k\cdot\Theta \|_{L^\infty(\R)} > 0	
\end{equation}
for almost every $(s,t) \in \Omega_0$.
If~$\Gamma$ and~$\Theta$ were smooth functions,
the inverse function theorem would immediately imply 
that $\mathscr{L}: \Omega_0 \to \Omega$
is a local smooth diffeomorphism, 
so that~$\Omega$ could be identified with 
the Riemannian manifold $(\Omega_0,g)$,
with~$\mathscr{L}$ realising an immersion in~$\R^{n+1}$.
Under our minimal regularity assumptions, however,
we have to be rather careful. 

\begin{proposition}\label{Prop.diffeo}
Suppose Assumption~\ref{Ass.Ass}.
Then $\mathscr{L}: \Omega_0 \to \Omega$ is a local 
$C^{0,1}$-diffeomorphism.
\end{proposition}
\begin{proof}
Given any bounded interval $I \subset \R$,
let $s_1,s_2 \in I$ and $t_1,t_2 \in (-a,a)$.
Let us look at the difference
$$
\begin{aligned}
  \mathscr{L}(s_2,t_2) - \mathscr{L}(s_1,t_1)
  &= \Gamma(s_2) - \Gamma(s_1) + N_\Theta(s_2)\, t_2 - N_\Theta(s_1)\, t_1
  \\
  &= \int_{s_1}^{s_2} T(\xi) \, \der\xi 
  + (t_2-t_1) \, N_\Theta(s_1) 
  + t_2 \int_{s_1}^{s_2} N_\Theta'(\xi) \, \der\xi
  \\
  &= \int_{s_1}^{s_2} 
  \Big(
  \big[1-t_2 \, (k\cdot\Theta)(\xi)\big] \, T(\xi) 
  + \big[t_2 \, (\Theta'\cdot N)(\xi)\big]
  \Big)
  \, \der\xi 
  + (t_2-t_1) \, N_\Theta(s_1) 
  \,,
\end{aligned}
$$
where we have used~\eqref{generic} and~\eqref{frame} 
and abbreviated $\Theta' \cdot N \deff \Theta_1' N_1 + \dots + \Theta_n' N_n$.
From the identity on the last line, 
we immediately conclude that
$$
  |\mathscr{L}(s_2,t_2) - \mathscr{L}(s_1,t_1)|
  \leq \big[(1+a \, \|k\cdot\Theta\|_\infty) + a \, \|\Theta'\|_\infty\big] \, |s_2-s_1|
  + |t_1-t_1|
  \,,
$$ 
where $\|\cdot\|_\infty$ denotes the supremum norm of $L^\infty(I)$
and $\|\Theta'\|_\infty \deff \||\Theta'|\|_\infty$.
Since the choice of the interval~$I$ has been arbitrary, 
we conclude that~$\mathscr{L}$ is a locally Lipschitz function.

To show that~$\mathscr{L}$ is a locally bi-Lipschitz function,
\ie~also the inverse~$\mathscr{L}^{-1}$ is a Lipschitz function
on $I \times (-a,a)$,
we further improve the expansions above to
$$
\begin{aligned}
  \Gamma(s_2) - \Gamma(s_1)
  &= T(s_1) \, (s_2-s_1)
  + \int_{s_1}^{s_2} \int_{s_1}^\xi (k\cdot N)(\eta) \, \der\eta \, \der\xi
  \,,
  \\
  \int_{s_1}^{s_2} N_\Theta'(\xi) \, \der\xi
  &= N(s_1) \cdot \int_{s_1}^{s_2} \Theta'(\xi) \, \der\xi
  - T(s_1) \int_{s_1}^{s_2} (k\cdot\Theta)(\xi) \, \der\xi
  \\
  &\quad 
  + \int_{s_1}^{s_2} \int_{s_1}^\xi \Theta'(\xi) \cdot N'(\eta)
  \, \der\eta \, \der\xi
  - \int_{s_1}^{s_2} \int_{s_1}^\xi (k\cdot\Theta)(\xi) \, T'(\eta)
  \, \der\eta \, \der\xi
  \,.
\end{aligned}
$$
That is, we can write
$
  \mathscr{L}(s_2,t_2) - \mathscr{L}(s_1,t_1)
  = A+B
$
with
$$
\begin{aligned}
  A &\deff T(s_1) \left[
  (s_2-s_1) - t_2\int_{s_1}^{s_2} (k\cdot\Theta)(\xi) \, \der\xi
  \right]
  + N(s_1) \cdot \left[
  (t_2-t_1) \, \Theta(s_1) 
  + t_2 \int_{s_1}^{s_2} \Theta'(\xi) \, \der\xi
  \right]
  ,
  \\
  B &\deff  
  \int_{s_1}^{s_2} \int_{s_1}^\xi (k\cdot N)(\eta) \, \der\eta \, \der\xi
  + t_2 \int_{s_1}^{s_2} \int_{s_1}^\xi \Theta'(\xi) \cdot N'(\eta)
  \, \der\eta \, \der\xi
  - t_2 \int_{s_1}^{s_2} \int_{s_1}^\xi (k\cdot\Theta)(\xi) \, T'(\eta)
  \, \der\eta \, \der\xi
  \,.
\end{aligned}
$$
For every $\delta_1,\delta_2 \in (0,1)$, we have
$$
\begin{aligned}
  |A|^2 
  &= \left[
  (s_2-s_1) - t_2\int_{s_1}^{s_2} (k\cdot\Theta)(\xi) \, \der\xi
  \right]^2
  + \left[
  (t_2-t_1) \, \Theta(s_1) 
  + t_2 \int_{s_1}^{s_2} \Theta'(\xi) \, \der\xi
  \right]^2
  \\
  &\geq  
  \delta_1 \, (s_2-s_1)^2
  - \frac{\delta_1}{1-\delta_1} \, a^2 \, \|k\cdot\Theta\|_\infty^2 \, (s_2-s_1)^2 
  + \delta_2 \, (t_2-t_1)^2 
  - \frac{\delta_2}{1-\delta_2} \, a^2 \, \|\Theta'\|_\infty^2 \, (s_2-s_1)^2 
  \\
  &=
  \delta_1 \, (s_2-s_1)^2 \, 
  \left[
  1 - \frac{1}{1-\delta_1} \, a^2 \, \|k\cdot\Theta\|_\infty^2
  - \frac{\delta_2/\delta_1}{1-\delta_2} \, a^2 \, \|\Theta'\|_\infty^2
  \right] 
  + \delta_2 \, (t_2-t_1)^2 
  \,.
\end{aligned}  
$$
By Assumption~\ref{Ass.Ass}, we can choose~$\delta_1$ 
so small that $1 - \frac{1}{1-\delta_1} \, a^2 \, \|k\cdot\Theta\|_\infty^2$
is positive. Then we choose~$\delta_2$ so small that
the square bracket in the last line above is positive.  
Altogether we can assure that there is a positive constant~$c$,
depending exclusively on the choice of the interval~$I$, 
such that 
$$
  |A| \geq c \sqrt{(s_2-s_1)^2 + (t_2-t_1)^2}
  \,.
$$
At the same time, using~\eqref{frame}, we have
$$
  |B| \leq (s_2-s_1)^2 
  \left(
  \|\kappa\|_\infty
  + a \, \|\Theta'\|_\infty \|\kappa\|_\infty
  + a \, \|\kappa\|_\infty^2
  \right)
  .
$$
Notice that $\|\kappa\|_\infty$ is finite by Assumption~\ref{Ass.Ass},
although~$\kappa$ is not supposed to be necessarily bounded on~$\R$.
That is, there exists a positive constant~$C$,
depending exclusively on the choice of the interval~$I$, 
such that $|B| \leq C |I| |s_2-s_1|$,
where~$|I|$ denotes the length of the interval~$I$.
Consequently,
$$
  |\mathscr{L}(s_2,t_2) - \mathscr{L}(s_1,t_1)|
  \geq (c-C |I|) \sqrt{(s_2-s_1)^2 + (t_2-t_1)^2}
$$
for all $s_1,s_2 \in I$ and $t_1,t_2 \in (-a,a)$.
Choosing~$|I|$ sufficiently small, we see that~$\mathscr{L}$
is invertible on $I \times (-a,a)$ and that the inverse~$\mathscr{L}^{-1}$
is a locally Lipschitz function.
\end{proof}

As a consequence of this proposition, 
the restriction~$\mathscr{L}\upharpoonright\Omega_0$
is a $C^{0,1}$-immersion.
Assuming additionally that $\mathscr{L} \upharpoonright \Omega_0$
is injective, then it is actually an embedding
and~$\Omega$ has a geometrical meaning of 
a non-self-intersecting strip.
For our purposes, however,
it is enough to assume that~$\Omega$ is an immersed submanifold.
Even less, disregarding the ambient space~$\R^{n+1}$ completely,
instead of~$\Omega$ we shall consider $(\Omega_0,g)$
as an abstract Riemannian manifold.
From now on, we thus assume the minimal hypotheses 
of Assumption~\ref{Ass.Ass} and nothing more.

\begin{remark}\label{Rem.spiral}
It is worth noticing that, contrary to the geodesic curvature~$k\cdot\Theta$,
the curvature~$\kappa$ is not assumed to be (globally) bounded
by Assumption~\ref{Ass.Ass}. In particular, $\Gamma$~is allowed to be 
a spiral with $\kappa(s) \to \infty$ as $s \to \pm\infty$.
\end{remark}

\subsection{The strip as a quantum Hamiltonian} 	
The word ``quantum'' refers to that we consider
the Hamiltonian of a free quantum particle constrained to~$\Omega$.
As usual, we model the Hamiltonian by the Laplace-Beltrami operator
in $\sii(\Omega)$, subject to Dirichlet boundary condition.
Since we think of~$\Omega$ as part of an abstract manifold
(not necessarily embedded in~$\R^{n+1}$),
we disregard the presence of extrinsic potentials 
occasionally added to the Laplace-Beltrami operator 
in order to justify quantisation on submanifolds (\cf~\cite{FrHe}).  

Using the identification $\Omega \cong (\Omega_0,g)$
with the metric~$g$ given by~\eqref{Eq metric}, 
the operator of our interest is thus
the self-adjoint operator~$H$ in the Hilbert space 
\begin{equation}\label{Hilbert}
  \Hilbert \deff \sii(\Omega_0,|g(s,t)|^{1/2}\,\der s \, \der t)
  = \sii(\Omega_0,f(s,t)\,\der s \, \der t)
\end{equation}
that acts as 
\begin{equation}\label{Eq LapBel strip}
  -\Delta_g \deff
  -|g|^{1/2} \partial_\mu |g|^{1/2} g^{\mu\nu} \partial_\nu 
  = - f^{-1} \, \partial_1 \, f^{-1}  \, \partial_1 
  - f^{-1} \, \partial_2 \, f  \, \partial_2 
\end{equation}
in~$\Omega_0$ 
and the functions in the operator domain vanish on~$\partial\Omega_0$.
Here we employ the standard notations 
$|g|\deff\det(g)$ and $(g^{\mu\nu})=g^{-1}$ 
together with the Einstein summation convention
with the range of indices being $\mu,\nu \in \{1,2\}$.
As usual we introduce~$H$ as the Friedrichs extension
of the operator $-\Delta_g$ initially defined on $C_0^\infty(\Omega_0)$. 
More specifically, $H$~is defined as the self-adjoint operator associated 
in~$\Hilbert$
(in the sense of the representation theorem~\cite[Thm.~VI.2.1]{Kato}) 
with the quadratic form
\begin{equation}\label{form}
\begin{aligned} 
  h[\psi] &\deff \|f^{-1/2}\partial_1\psi\|_\Hilbert^2 
  + \|f^{1/2}\partial_2\psi\|_\Hilbert^2
  = \int_{\Omega_0} \frac{|\partial_1\psi(s,t)|^2}{f(s,t)} \, \der s \, \der t
  + \int_{\Omega_0} |\partial_2\psi(s,t)|^2 \, f(s,t) \, \der s \, \der t
  \,, 
  \\
  \Dom(h) &\deff \overline{C_0^\infty(\Omega_0)}^{\|\cdot\|_{\Hilbert_1}}
  \,, \qquad \mbox{where} \qquad
  \|\psi\|_{\Hilbert_1} \deff \sqrt{h[\psi]+\|\psi\|_\Hilbert^2}
  \,.
\end{aligned} 
\end{equation}
By~$\Hilbert_1$ we shall understand the Hilbert space~$\Dom(h)$
equipped with the norm~$\|\cdot\|_{\Hilbert_1}$.

Under our standing hypotheses of Assumption~\ref{Ass.Ass},
the crucial bound~\eqref{Ass.bound} holds 
and, moreover, the function~$f$ is locally bounded.
Consequently, one has 
\begin{equation}\label{compact}
  \{\psi \in W_0^{1,2}(\Omega_0) : \ 
  \supp\psi \subset [-R,R]\times[-a,a] \mbox{ for some } R > 0
  \}
  \subset \Dom(h)
  \,.
\end{equation}
Assuming in addition that
\begin{equation}\label{Ass2}
  |\Theta'| \in L^\infty(\R) \,,   
\end{equation}
then there exists a positive constant~$C$ such that
even the global bounds 
\begin{equation}\label{global}
  C^{-1} \leq f(s,t) \leq C
\end{equation}
hold for almost every $(s,t) \in \Omega_0$.
Consequently, $\|\cdot\|_{\Hilbert_1}$ is equivalent to the usual norm
of the Sobolev space $W^{1,2}(\Omega_0)$ 
and one has $\Dom(h) = W_0^{1,2}(\Omega_0)$.
In this paper, however, we proceed in a greater generality
without assuming the extra hypothesis~\eqref{Ass2}.

%-----------------------------------% 
\section{Asymptotically flat strips}\label{Sec.ess}
%-----------------------------------%
%
If the strip~$\Omega$ is \emph{flat} in the sense that 
its metric~\eqref{Eq metric} is Euclidean, 
\ie~$f=1$ (identically),
then~$H$ coincides with the Dirichlet Laplacian in~$\Omega_0$,
which we denote here by~$H_0$.
More specifically, $H_0$~is the operator in $\Hilbert_0\deff\sii(\Omega_0)$
associated with the quadratic form
$h_0[\psi] \deff \int_{\Omega_0} |\nabla\psi(x)|^2\, \der x$,
$\Dom(h_0)\deff W_0^{1,2}(\Omega_0)$.
It is well known that  
$$
  \sigma(H_0) = [E_1,\infty) 
  \qquad \mbox{with} \qquad
  E_1 \deff \left(\frac{\pi}{2a}\right)^2
$$ 
and that the (purely essential) spectrum 
is in fact purely absolutely continuous.

In this section, we consider quantum strips which 
are \emph{asymptotically flat} in the sense that 
their metric~\eqref{Eq metric} converges 
to the flat metric at the infinity of~$\Omega_0$. 
More specifically, we impose the conditions
\begin{equation}\label{vanish}
  \lim_{|s|\to\infty} (k \cdot \Theta)(s) = 0
  \qquad \mbox{and} \qquad
  \lim_{|s|\to\infty} |\Theta'(s)| = 0 
  \,.
\end{equation}
Since quantum propagating states are expected to be determined 
by the behaviour of the metric at infinity, 
the following result is very intuitive.

\begin{theorem}\label{Thm.ess}
Suppose Assumption~\ref{Ass.Ass}.
If~\eqref{vanish} holds, then   
\begin{equation*}
  \espec (H) = [E_1,\infty)
  \,.
\end{equation*}
\end{theorem}

We establish the theorem as a consequence of two lemmata.
First we show that the energy of the propagating states
cannot descend below~$E_1$.   

\begin{lemma}\label{Lem.ess.lower}
Suppose Assumption~\ref{Ass.Ass}.
If~\eqref{vanish} holds, then   
\begin{equation*}
\inf \espec (H) \geq E_1 \,.
\end{equation*}
\end{lemma}
\begin{proof}
Given any arbitrary positive number~$s_0$,
we divide~$\Omega_0$ into an interior and an exterior part 
with respect to~$s_0$ as follows:
\begin{equation*}
  \Omega_{0,\mathrm{int}} \deff (-s_0, s_0) \times (-a,a) 
  \,, \qquad
  \Omega_{0,\mathrm{ext}} \deff \Om_0 \setminus \overline{\Omega_{0,\mathrm{int}}} 
  \,.
\end{equation*}
Imposing Neumann boundary conditions on the segments $\{\pm s_0\} \times (-a,a)$,
one gets the lower bound 
\begin{equation}\label{bracketing}
  H \geq H^N \deff H^N_{\mathrm{int}} \oplus H^N_{\mathrm{ext}}
\end{equation}
in the sense of forms in~$\Hilbert$.
Here~$H^N_{\mathrm{int}}$ is the operator in 
$\Hilbert_\mathrm{int}\deff\sii(\Omega_{0,\mathrm{int}},f(s,t)\,\der s \, \der t)$
associated with the quadratic form~$h_\mathrm{int}^N$ that acts as~$h$
but whose domain is restricted to $\Omega_{0,\mathrm{int}}$.
More specifically, 
$$
\begin{aligned}
  h_\mathrm{int}^N[\psi] &\deff  
  \int_{\Omega_{0,\mathrm{int}}} 
  \frac{|\partial_1\psi(s,t)|^2}{f(s,t)} \, \der s \, \der t
  + \int_{\Omega_{0,\mathrm{int}}} 
  |\partial_2\psi(s,t)|^2 \, f(s,t) \, \der s \, \der t
  \,, 
  \\
  \Dom(h_\mathrm{int}^N) &\deff \left\{
  \psi \deff \tilde{\psi} \upharpoonright \Omega_{0,\mathrm{int}} : \
  \tilde{\psi} \in \Dom(h) 
  \right\}
  \,. 
\end{aligned}
$$
Note that no boundary conditions are imposed on the parts 
$\{\pm s_0\} \times (-a,a)$ of the boundary $\partial\Omega_{0,\mathrm{int}}$
in the form domain, while Dirichle boundary conditions 
are considered on the remaining parts of the boundary.
The operator~$H^N_{\mathrm{ext}}$, the form~$h^N_{\mathrm{ext}}$ 
and the Hilbert space $\Hilbert_\mathrm{ext}$ are defined analogously.

Employing the Neumann bracketing described above, 
we have 
\begin{equation*}
  \inf \espec (H) 
  \geq \inf \espec (H^N) 
  = \inf\espec(H^N_{\mathrm{ext}})
  \geq \inf\sigma(H^N_{\mathrm{ext}})
  \,.
\end{equation*}
Here the first inequality follows from~\eqref{bracketing}  
via the minimax principle,
the equality is due to the fact that 
the spectrum of $H^N_{\mathrm{int}}$ is purely discrete
and the last inequality is trivial.	
Hence, it is sufficient to find a suitable lower bound 
to the spectrum of $H^N_{\mathrm{ext}}$.
To this aim, for every $\psi \in \Dom(h_\mathrm{ext}^N)$, 
we estimate the quadratic form as follows:
\begin{align*}
  h^N_{\mathrm{ext}}[\psi] 
  &\geq \int_{\Omega_{0,\mathrm{ext}}} 
  |\partial_2\psi(s,t)|^2 \, f(s,t) \, \der s \, \der t
  \\
  &\geq \big(\essinf_{\Omega_{0,\mathrm{ext}}} f\big)
  \int_{\Omega_{0,\mathrm{ext}}} 
  |\partial_2\psi(s,t)|^2 \, \der s \, \der t
  \\
  &\geq E_1 \, \big(\essinf_{\Omega_{0,\mathrm{ext}}} f\big)
  \int_{\Omega_{0,\mathrm{ext}}} 
  |\psi(s,t)|^2 \, \der s \, \der t
  \\
  &\geq E_1 \, \big(\essinf_{\Omega_{0,\mathrm{ext}}} f\big) 
  \, \big(\esssup_{\Omega_{0,\mathrm{ext}}} f\big)^{-1}
  \int_{\Omega_{0,\mathrm{ext}}} 
  |\psi(s,t)|^2 \, f(s,t) \, \der s \, \der t
  \\
  &= E_1 \, \big(\essinf_{\Omega_{0,\mathrm{ext}}} f\big) 
  \, \big(\esssup_{\Omega_{0,\mathrm{ext}}} f\big)^{-1}
  \|\psi\|^2_{\Hilbert_{\mathrm{ext}}} \, .
\end{align*}
Consequently,
$$
  \inf\sigma_\mathrm{ess}(H)
  \geq E_1 \, \big(\essinf_{\Omega_{0,\mathrm{ext}}} f\big) 
  \, \big(\esssup_{\Omega_{0,\mathrm{ext}}} f\big)^{-1}
  \,.
$$
Taking the limit $s_0 \to \infty$,   
the asymptotic hypothesis~\eqref{vanish} ensures 
that the right-hand side tends to~$E_1$,
while the left-hand side is independent of~$s_0$. 
This concludes the proof of the lemma. 
\end{proof}

It remains to show that all energies above~$E_1$ belong 
to the essential spectrum.
\begin{lemma}\label{Lem.ess.upper}
Suppose Assumption~\ref{Ass.Ass}.
If~\eqref{vanish} holds, then   
\begin{equation*}
  \espec (H) \supset [E_1,\infty) \,.
\end{equation*}
\end{lemma}
\begin{proof}
Our argument is based on the Weyl criterion adapted to quadratic forms
(see~\cite[Thm.~5]{KL} for the proof of the criterion 
and~\cite[Lem.~5.3]{KKriz} for an original application to quantum waveguides).  
It states that to prove that~$\eta$ is in the spectrum of the operator~$H$, 
it is enough to find a sequence 
$\lbrace \psi_n \rbrace _{n\in\N} \subset \Dom\left(h\right)$ 
such that
\begin{enumerate}
\item[(i)] 
$
\displaystyle
\liminf_{n\to\infty} \| \psi_n \|_{\Hilbert} > 0
$, 
\item[(ii)] 
$
\displaystyle
\lim_{n\to\infty} 
\| \left(H-\eta\right) \psi_n \|_{\Hilbert_{-1}} 
= 0$,
\end{enumerate}
where $\Hilbert_{-1}$ denotes the dual space of $\Hilbert_1$. 
Notice that the mapping $H+1 : \Hilbert_1 \to \Hilbert_{-1}$ 
is an isomorphism and that the dual norm is given by
\begin{equation*} 
  \| \psi \|_{\Hilbert_{-1}} 
  = \sup_{\stackrel[\phi \not= 0]{}{\phi \in \Dom\left(h\right)}}
  \frac{|\tensor[_{\Hilbert_1}]{(\phi,\psi)}{_{\Hilbert_{-1}}}|}
  {\|\phi\|_{\Hilbert_1}}
  \,,
\end{equation*}
where  
$\tensor[_{\Hilbert_1}]{( \cdot, \cdot )}{_{\Hilbert_{-1}}}$
denotes the duality pairing between~$\Hilbert_1$ and~$\Hilbert_{-1}$.
In contrast to the conventional Weyl criterion,
the advantage of this characterisation is that the sequence
is required to lie in the form domain only 
and the limit in~(ii) is taken in the weaker topology of~$\Hilbert_{-1}$.

We parameterise~$\eta$ by setting 
$\eta = \lambda^2 + E_1 $ with any $\lambda \in \R$. 
Note that the differential equations $-\Delta\psi=\eta\psi$
is solved by $(s,t) \mapsto \chi_1\left(t\right) \mathrm{e}^{i \lambda s}$,
where
\begin{equation}\label{chi1}
  \chi_1(t)  \deff \sqrt{\frac{1}{a}} \, \cos\big(\sqrt{E_1} t\big)  
\end{equation}
denotes the normalised eigenfunction of the Dirichlet Laplacian
in~$(-a,a)$ corresponding to the eigenvalue~$E_1$.
However, this solution does not even belong to~$\Hilbert$.
To get an approximate solution 
which simultaneously belongs to $\Dom(h)$ 
and is ``localised at infinity'', 
for every $n \in \N \deff \{1,2,\dots\}$ we set
\begin{equation*}
  \varphi_n\left(s\right)  
  \deff \frac{1}{\sqrt{n}} \, \varphi \left( \frac{s}{n} - n \right) ,
\end{equation*}
where $\varphi \in C_0^\infty(\R)$ is any function 
such that $\supp \varphi \subset \left(-1,1\right)$
and $\|\varphi\|_{\sii(\R)}=1$.
The normalisation factor is chosen in such a way that 
\begin{equation}\label{varphis}
  \|\varphi_n\|_{\sii(\R)} = \|\varphi\|_{\sii(\R)} = 1
  \,, \qquad
  \|\varphi_n'\|_{\sii(\R)} = n^{-1} \, \|\varphi'\|_{\sii(\R)} 
  \,, \qquad
  \|\varphi_n''\|_{\sii(\R)} = n^{-2} \, \|\varphi''\|_{\sii(\R)} 
  \,.
\end{equation}
Notice also that $\supp \varphi _n \subset \left(n^2-n, n^2 + n\right)$.
We then define
\begin{equation*}
  \psi_n\left(s,t\right) 
  \deff \varphi_n\left(s\right) \chi_1\left(t\right) \mathrm{e}^{i \lambda s}
  \,.
\end{equation*}
Recalling~\eqref{compact}, 
we clearly have $\psi_n \in \Dom(h)$ for every $n \in \N$.
Our aim is to show that $\lbrace \psi _n \rbrace_{n\in\N}$ 
satisfies conditions~(i) and~(ii) of the modified Weyl criterion.
 
First of all, notice that, 
due to~\eqref{Ass.bound} 
and the normalisations of~$\varphi$ and~$\chi_1$, 
we have 
$$
  \| \psi_n \|_\Hilbert^2
  \geq 1 - a \, \|k\cdot\Theta\|_{L^\infty(\R)} > 0
  \,,
$$
so the condition~(i) clearly holds. 
Next, for every $\phi \in \Dom(h)$, we have 
$$
\begin{aligned}
  |\tensor[_{\Hilbert_1}]{(\phi,(H-\eta)\psi_n)}{_{\Hilbert_{-1}}}|
  &= |h(\phi,\psi_n) - \eta (\phi,\psi_n)_{\Hilbert}|
  \\
  &\leq |h_1(\phi,\psi_n) - \lambda^2 (\phi,\psi_n)_{\Hilbert}|
  + |h_2(\phi,\psi_n) - E_1 (\phi,\psi_n)_{\Hilbert}|
  \,,
\end{aligned}  
$$
where, recalling~\eqref{form}, 
$h_1[\psi] \deff \|f^{-1/2}\partial_1\psi\|_\Hilbert^2$,
$h_2[\psi] \deff \|f^{1/2}\partial_2\psi\|_\Hilbert^2$,
$\Dom(h_1) \deff \Dom(h) \deffin \Dom(h_2)$. 
Integrating by parts and using that $-\chi_1''=E_1\chi_1$
together with the normalisations of~$\varphi$ and~$\chi_1$, 
we have
\begin{equation}\label{h1}
\begin{aligned}
  |h_2(\phi,\psi_n) - E_1 (\phi,\psi_n)_{\Hilbert}|
  &= \left|
  \int_{\Omega_0} \bar{\phi}(s,t) \, \partial_2\psi_n(s,t) \, \partial_2 f(s,t)
  \, \der s \, \der t
  \right|
  \\
  &\leq \|\phi\|_{\Hilbert} \, \|\partial_2\psi_n\|_{\Hilbert_0} \,
  \left\|\frac{\partial_2 f}{\sqrt{f}}\right\|_{\infty, n}
  \\
  &\leq \|\phi\|_{\Hilbert_1} \, \sqrt{E_1} \,
  \left\|\frac{\partial_2 f}{\sqrt{f}}\right\|_{\infty, n}
  \,,
\end{aligned}  
\end{equation}
where 
$
  \|\cdot\|_{\infty,n} 
  \deff \|\cdot\|_{L^\infty(\supp\varphi_n \times (-a,a))}
$.
At the same time, we have
$$
\begin{aligned}
  h_1(\phi,\psi_n) 
  &= 
  \int_{\Omega_0} \partial_1\bar{\phi}(s,t) \, \partial_1\psi_n(s,t) 
  \left[\frac{1}{f(s,t)}-1\right]
  \der s \, \der t
  - \int_{\Omega_0} \bar{\phi}(s,t) \, \partial_1^2\psi_n(s,t) 
  \, \der s \, \der t
  \,,
  \\
  (\phi,\psi_n)_{\Hilbert} 
  &= \int_{\Omega_0} \bar{\phi}(s,t) \, \psi_n(s,t) 
  \left[f(s,t)-1\right]
  \der s \, \der t
  + \int_{\Omega_0} \bar{\phi}(s,t) \, \psi_n(s,t) \, \der s \, \der t
  \,.
\end{aligned}  
$$
Consequently, using that  
$
  -\partial_1^2\psi_n(s,t) - \lambda^2 \psi_n(s,t)
  = [-\varphi_n''(s) - 2 i \lambda \varphi_n'(s)] \, 
  \mathrm{e}^{i\lambda s} \, \chi_1(t)
$
and the normalisations of~$\varphi$ and~$\chi_1$ again,
we get
\begin{equation}\label{h2}
\begin{aligned}
  |h_1(\phi,\psi_n) - \lambda^2 (\phi,\psi_n)_{\Hilbert}| 
  &\leq  \|f^{-1/2}\phi\|_{\Hilbert} \, \|\partial_1\psi_n\|_{\Hilbert_0} \,
  \left\|\frac{1}{\sqrt{f}}-\sqrt{f}\right\|_{\infty, n}
  \\
  &\quad + \lambda^2 \|\phi\|_{\Hilbert} \, \|\psi_n\|_{\Hilbert_0} \,
  \left\|\sqrt{f}-\frac{1}{\sqrt{f}}\right\|_{\infty, n}
  \\
  &\quad + \|\phi\|_{\Hilbert} \, \|\varphi_n''+ 2 i \lambda \varphi_n'\|_{\sii(\R)} \,
  \left\|\frac{1}{\sqrt{f}}\right\|_{\infty, n}
  \\
  &\leq \|\phi\|_{\Hilbert_1} \, \|\varphi_n'+i\lambda\varphi_n\|_{\sii(\R)} \,
  \left\|\frac{1}{\sqrt{f}}-\sqrt{f}\right\|_{\infty, n}
  \\
  &\quad + \lambda^2 \|\phi\|_{\Hilbert_1} \, 
  \left\|\sqrt{f}-\frac{1}{\sqrt{f}}\right\|_{\infty, n}
  \\
  &\quad + \|\phi\|_{\Hilbert_1} \, \|\varphi_n''+ 2 i \lambda \varphi_n'\|_{\sii(\R)} \,
  \left\|\frac{1}{\sqrt{f}}\right\|_{\infty, n}
  \,.
\end{aligned}  
\end{equation}
Putting~\eqref{h1} and~\eqref{h2} together, 
we finally arrive at
$$
\begin{aligned}
  \|(H-\eta)\psi_n\|_{\Hilbert_{-1}} 
  &\leq \sqrt{E_1} \,
  \left\|\frac{\partial_2 f}{\sqrt{f}}\right\|_{\infty, n}
  + \big(\|\varphi_n'+i\lambda\varphi_n\|_{\sii(\R)}+\lambda^2\big)
  \left\|\frac{1}{\sqrt{f}}-\sqrt{f}\right\|_{\infty, n}
  \\
  &\quad+ \|\varphi_n''+ 2 i \lambda \varphi_n'\|_{\sii(\R)} \,
  \left\|\frac{1}{\sqrt{f}}\right\|_{\infty, n}
  \,.
\end{aligned} 
$$
Here the first line on the right-hand side tends to zero as $n \to \infty$
due to~\eqref{vanish}, while the second line vanishes as $n \to \infty$
due to~\eqref{varphis}.
This establishes~(ii) and the lemma is proved. 
\end{proof}

Theorem~\ref{Thm.ess} follows as a direct consequence 
of Lemmata~\ref{Lem.ess.lower} and~\ref{Lem.ess.upper}.

%---------------------------%
\section{Purely bent strips}\label{Sec.bent}
%---------------------------%
%
In this section, we consider strips constructed in such a way 
that the twisting vector~$\Theta$ is constant, 
so that~$N_\Theta$ is relatively parallel and~$\Omega$ is untwisted.
We show that the geodesic curvature $k\cdot\Theta$ 
acts as an attractive interaction in the sense that 
it diminishes the spectrum.
Recall that $k\cdot\Theta$ can be equal to zero even if $\kappa \not=0$
(like, for instance, in Figure~\ref{Fig.helix}, right). 

\begin{theorem}\label{Thm.bent}
Suppose Assumption~\ref{Ass.Ass}. 
If $\Theta' = 0$ and $k\cdot \Theta \not= 0$, then 
\begin{equation*}
\inf \sigma (H) < E_1 \,.
\end{equation*}
\end{theorem}
\begin{proof}
The proof is based on the variational strategy of finding 
a trial function $\psi \in \Dom(h)$ such that
\begin{equation}\label{h.shifted}
h_1[\psi] \deff h[\psi] - E_1 \| \psi \|_\Hilbert^2 < 0 \,. 
\end{equation}
Following~\cite{K1}, we shall achieve the strict inequality
by mollifying the first transverse eigenfunction~$\chi_1$
introduced in~\eqref{chi1}.

Let $\varphi_1 \in C_0^\infty(\R)$ be a real-valued function 
such that $0 \leq \varphi_1 \leq 1$, $\varphi_1=1$ on $[-1,1]$
and $\varphi_1=0$ on $\R \setminus [-2,2]$.
Setting $\varphi_n(s) \deff \varphi_1(s/n)$ for every $n \in \N$, 
we get a family of functions from $W^{1,2}(\R)$ such that
$\varphi_n \to 1$ pointwise as $n \to \infty$ and 
\begin{equation}\label{parabolic}
  \|\varphi_n'\|_{\sii(\R)}^2 
  = n^{-1} \, \|\varphi_1'\|_{\sii(\R)}^2 
  \xrightarrow[n \to \infty]{}
  0
  \,.
\end{equation}
Defining
\begin{equation*}%\label{Eq Bent strip psi_n}
  \psi_n(s,t) \deff \varphi_n(s) \, \chi_1 (t) 
\end{equation*}
and integrating by parts with help of $-\chi_1''=E_1\chi_1$,
we have
$$
  h_1[\psi_n] = 
  \int_{\Omega_0} \frac{|\partial_1\psi_n(s,t)|^2}{f(s,t)} \, \der s \, \der t
  + \frac{1}{2} \int_{\Omega_0} |\psi_n(s,t)|^2 \, \partial_2^2 f(s,t) \, \der s \, \der t
  = \int_{\Omega_0} \frac{|\varphi_n'(s)|^2|\chi_1(t)|^2}{f(s,t)} \, \der s \, \der t
  \,.
$$
Here the second equality follows from the fact 
that the Jacobian~$f$ of the metric~\eqref{Eq metric} reduces to
\begin{equation*}
  f(s,t) = 1 -t \, k(s) \cdot \Theta(s) 
\end{equation*}
provided that~$\Theta$ is constant, 
so it is linear in the second variable and $\partial_2^2f=0$.
Using~\eqref{Ass.bound} and~\eqref{parabolic}, 
we therefore conclude that
\begin{equation}\label{step1}
  \lim_{n\to\infty} h_1[\psi_n] = 0
  \,.
\end{equation}

It follows that the functional~$h_1$ vanishes at~$\chi_1$ 
in a generalised sense. The next (and last) step in our strategy is to
show that~$\chi_1$ does not correspond to the minimum of the functional.
To this purpose, we add the following asymmetric perturbation 
\begin{equation*}
  \psi_{n,\eps} (s,t) \deff \psi_n (s,t) + \eps \, \phi(s,t)
  \,, \qquad \mbox{where} \qquad
  \phi(s,t) \deff \eta(s)\, t \, \chi_1 (t) \,, 
\end{equation*}
with $\eps \in \R$ and $\eta \in C_0^\infty(\R)$ being 
a non-zero real-valued function to be specified later.
Plugging it into the functional, we obviously have
\begin{equation}\label{q-form}
  h_1[\psi_{n,\eps}] 
  = h_1[\psi_n] + 2 \, \eps \, h_1(\psi_n,\phi) + \eps^2 \, h_1[\phi]
  \,.
\end{equation}
Since $\varphi_n=1$ on $\supp\eta$ for all sufficiently large~$n$,
the central term is in fact independent of~$n$ and equals
$$
\begin{aligned}
  h_1(\psi_n,\phi) 
  &= 
  \int_{\Omega_0} \eta(s) \, \chi_1'(t) \, [t\,\chi_1(t)]' \, f(s,t) \, \der s \, \der t
  - E_1 \int_{\Omega_0} \eta(s) \, t \, |\chi_1(t)|^2 \, f(s,t) \, \der s \, \der t
  \\
  &= -\int_{\Omega_0} 
  \eta(s) \, \chi_1'(t) \, t\chi_1(t) \, \partial_2 f(s,t) \, \der s \, \der t
  \\
  &= \frac{1}{2} \int_{\Omega_0} 
  \eta(s) \, |\chi_1(t)|^2 \, \partial_2 f(s,t) \, \der s \, \der t
  = -\frac{1}{2} \int_{\R} 
  \eta(s) \, (k\cdot\Theta)(s) \, \der s  
  \,.
\end{aligned}
$$
Here the second and third equalities follow by integrations by parts
using that $-\chi_1''=E_1\chi_1$ and $\partial_2^2 f = 0$.
Since $k\cdot\Theta$ is not identically equal to zero 
by the hypothesis of the theorem, it is possible to choose~$\eta$
in such a way that the last integral is positive.
Summing up, $h_1(\psi_n,\phi)$ equals a negative number 
for all sufficiently large~$n$.   
Coming back to~\eqref{q-form}, it is thus possible to choose
a positive~$\eps$ so small that sum of the last two terms
on the right-hand side of~\eqref{q-form} is negative. 
Then, recalling~\eqref{step1}, 
we can choose~$n$ so large that $h_1[\psi_{n,\eps}]<0$.	 
Hence, $\inf\sigma(H) < 0$ by the Rayleigh-Ritz variational characterisation
of the lowest point in the spectrum of~$H$.
\end{proof}

As a consequence of Theorem~\ref{Thm.bent}, 
if the strip is in addition asymptotically flat
in the sense of~\eqref{vanish} (of course, just the first limit
is relevant under the hypotheses of Theorem~\ref{Thm.bent}),
then the essential spectrum starts by~$E_1$ (\cf~Theorem~\ref{Thm.ess})
and the spectral threshold $\inf\sigma(H)$
necessarily corresponds to a discrete eigenvalue. 

\begin{corollary} 
In addition to the hypotheses of Theorem~\ref{Thm.bent},
let us assume~\eqref{vanish}. Then 
\begin{equation*}
  \sigma_\mathrm{disc}(H) \cap (0,E_1) \not= \varnothing
  \,.
\end{equation*}
\end{corollary}

This is a generalisation of the celebrate result~\cite{ES} 
about the existence of quantum bound states in curved planar quantum waveguides.

%------------------------------%
\section{Purely twisted strips}\label{Sec.twist}
%------------------------------%
%
In this section, we consider strips constructed in such a way 
that $k \cdot \Theta=0$, so that~$\Omega$ is unbent.
%see Figure~\ref{Fig.twist}.
Recall that this hypothesis does not necessarily mean 
that~$\Gamma$ is a straight line
(for instance, the setting in Figure~\ref{Fig.helix}, right, is admissible).
We show that the twisting vector~$\Theta$ 
acts as a repulsive interaction in the sense that 
it induces Hardy-type inequalities whenever~$\Theta'$
is not identically equal to zero (but it is not too large).

\begin{theorem}\label{Thm.Hardy}
Suppose Assumption~\ref{Ass.Ass}. 
If $k \cdot \Theta = 0$ and $\Theta' \not= 0$ satisfies
\begin{equation}\label{open}
  a \, \| \dvtheta \|_{L^\infty(\R)} \leq \sqrt{2}
  \,,
\end{equation}
then there exists a positive constant~$c$ 
such that the inequality 
\begin{equation}\label{Hardy.global}
  H - E_1 \geq c \, \rho
\end{equation}
holds in the sense of quadratic forms in~$\Hilbert$,
where $\rho(s,t) \deff 1/(1 + s^2)$. 
\end{theorem}

To prove the theorem, we follow the strategy of~\cite{K3}.
The main idea is to introduce a function $\lambda : \R \to \R$ by setting
\begin{equation}\label{lambda}
  \lambda (s) 
  \deff 
  \inf_{\stackrel[\psi \not= 0]{}{\psi \in W_0^{1,2}((-a,a))}}
  \frac{\displaystyle \int_{-a}^a |\psi'(t)|^2 \, f(s,t) \, \dd t}
  {\displaystyle \int_{-a}^a | \psi(t)|^2 \, f(s,t) \, \dd t} - E_1 
  \,.
\end{equation}
We keep the same notation for the function 
$ \lambda \otimes 1$ on $\R \times (-a,a)$.
Note that under the assumption $k\cdot\Theta=0$ of this section, 
the Jacobian~$f$ of the metric~\eqref{Eq metric} reduces to
\begin{equation}\label{metric.unbent}
  f(s,t) = \sqrt{1 + t^2 \, \Theta'(s)^2 } 
  \,.
\end{equation}
The following lemma is the crucial ingredient 
in the proof of Theorem~\ref{Thm.Hardy}.  

\begin{lemma}\label{Lem.lambda}
Under the assumptions of Theorem~\ref{Thm.Hardy}, 
$\lambda$ is a non-negative non-trivial function.
\end{lemma}
\begin{proof}
Fix any $s \in \R$.
Employing the change of the test function 
$\phi \deff \sqrt{f} \, \psi$ and by integrating by parts, 
we obtain
\begin{equation}\label{lambda.bis}
  \lambda(s) = 
  \inf_{\stackrel[\phi \not= 0]{}{\phi \in W_0^{1,2}((-a,a))}}
  \frac{\displaystyle \int_{-a}^a \left( |\phi'(t)|^2 
  - E_1 \, | \phi (t)| ^2 
  + V(s,t) \, |\phi (t) |^2 \right) \dd t}
  {\displaystyle \int_{-a}^a | \phi(t)|^2 \, \dd t}
\end{equation}
with
\begin{equation*}
  V(s,t) \deff 
  \frac{\Theta'(s)^2 \left( 2 - t^2 \, \Theta'(s)^2 \right) }
  {4 \, f(s,t)^4}
  \,.
\end{equation*}
We note that~$\lambda(s)$ is the spectral threshold of the self-adjoint
operator~$L$ in $\sii((-a,a))$ associated with the closed form
$$
  l[\phi]\deff\int_{-a}^a \left( |\phi'(t)|^2 
  - E_1 \, | \phi (t)| ^2 
  + V(s,t) \, |\phi (t) |^2 \right) \dd t
  \,,
  \qquad
  \Dom(l) \deff W_0^{1,2}((-a,a))
  \,.
$$
Since the resolvent of~$L$ is compact, 
$\lambda(s)$ is the lowest eigenvalue of~$L$.
Let us denote by~$\phi_1$ a corresponding eigenfunction.
By standard arguments (\cf~\cite[Thm.~8.38]{Gilbarg-Trudinger}),
the eigenvalue is simple and~$\phi_1$ can be chosen to be positive. 
The infimum in~\eqref{lambda.bis} is clearly achieved by~$\phi_1$.
Due to the hypothesis~\eqref{open}, the function~$V$ is non-negative. 
At the same time, one has the Poincar\'e inequality 
\begin{equation}\label{Poincare}
  \forall \phi \in W_0^{1,2}((-a,a)) 
  \,, \qquad
  \int_{-a}^a |\phi'(t)|^2 \, \der t
  \geq E_1 \int_{-a}^a |\phi(t)|^2 \, \der t
  \,.
\end{equation}
Consequently, $\lambda(s)$ is clearly non-negative.
Now, assume that $\lambda(s)=0$. 
Then necessarily $V(s,t)=0$ for every $t\in(-a,a)$,
which implies $\Theta'(s)=0$.
Since $\Theta'$ is supposed not to be identically equal to zero,
we necessarily have $\lambda\not=0$ as well.
\end{proof}

Using just the definition~\eqref{lambda} in~\eqref{form},
we immediately get the inequality
\begin{equation}\label{Hardy.local}
  H - E_1 \geq \lambda
  \,.
\end{equation}
By Lemma~\ref{Lem.lambda}, it is a Hardy-type inequality
whenever the assumptions of Theorem \ref{Thm.Hardy} hold true.
We call it a \emph{local} Hardy inequality because 
the defect of~\eqref{Hardy.local} is that the right-hand side 
might not be positive everywhere in~$\Omega_0$
(like, for instance, if~$\Theta'$ is compactly supported).
To transfer it into the \emph{global} Hardy inequality~\eqref{Hardy.global} 
of Theorem~\ref{Thm.Hardy},
we use the longitudinal kinetic energy that we have neglected 
when deriving~\eqref{Hardy.local}. 

\begin{proof}[Proof of Theorem~\ref{Thm.Hardy}]
Let $\psi \in C_0^\infty(\Omega_0)$. 
Under the hypotheses of the theorem, 
it follows from Lemma~\ref{Lem.lambda} 
that~$\lambda$ is non-negative and non-trivial. 
Let us fix any bounded open interval $I \subset \R$
on which~$\lambda$ is non-trivial.   
Let us abbreviate $\Omega_0^I \deff I \times (-a,a)$
and $\Hilbert_I \deff \sii(\Omega_0^I,f(s,t)\,\der s \, \der t)$.
We shall widely use the bounds
$$
  1 \leq f(s,t) \leq \sqrt{1+a^2 \|\Theta'\|_{L^\infty(I)}^2} \deffin C
$$
valid for almost every $(s,t) \in \Omega_0^I$.
Recall the definition of the shifted form~$h_1$ given in~\eqref{h.shifted}.
By using the definition~\eqref{lambda} in~\eqref{form}, we get 
\begin{equation}\label{writing0}
\begin{aligned}
  h_1[\psi]  
  &\geq
  \int_{\Omega_0} \frac{|\partial_1\psi(s,t)|^2}{f(s,t)} \, \der s \, \der t
  + \int_{\Omega_0} \lambda(s) \, |\psi(s,t)|^2\, f(s,t) \, \der s \, \der t
  \\
  &\geq 
  \int_{\Omega_0^I} \frac{|\partial_1\psi(s,t)|^2}{f(s,t)} \, \der s \, \der t
  + \int_{\Omega_0^I} \lambda(s) \, |\psi(s,t)|^2\, f(s,t) \, \der s \, \der t
  \\
  &\geq 
  C^{-1}
  \int_{\Omega_0^I} 
  \left( |\partial_1\psi(s,t)|^2 + \lambda(s) \, |\psi(s,t)|^2 \right)
  \der s \, \der t
  \\
  &\geq 
  C^{-1} \lambda_0 
  \int_{\Omega_0^I} |\psi(s,t)|^2 \, \der s \, \der t
  \\
  &\geq 
  C^{-2}\lambda_0 
  \, \|\psi\|_{\Hilbert_I}^2 
  \,,
\end{aligned}
\end{equation}
where $\lambda_0 > 0$ is the lowest eigenvalue 
of the Schr\"odinger operator $-\partial_s^2 + \lambda(s)$ in $\sii(I)$,
subject to Neumann boundary conditions.

Let us denote by~$s_0$ the middle point of~$I$. 
Let $\eta \in C^\infty(\R)$ be such that $0 \leq \eta \leq 1$,
$\eta=0$ in a neighbourhood of~$s_0$ and $\eta=1$ outside~$I$. 
Let us denote by the same symbol~$\eta$ the function 
$\eta \otimes 1$ on $\R \times (-a,a)$,
and similarly for its derivative~$\eta'$. 
Writing $\psi = \eta\psi + (1-\eta)\psi$, 
we have
\begin{equation}\label{writing} 
\begin{aligned}
  \int_{\Omega_0} 
  \frac{|\psi(s,t)|^2}{1+(s-s_0)^2} \, \der s \, \der t
  &\leq 2 \int_{\Omega_0} 
  \frac{|(\eta\psi)(s,t)|^2}{(s-s_0)^2} \, \der s \, \der t 
  + 2 \int_{\Omega_0} 
  |((1-\eta)\psi)(s,t)|^2\, \der s \, \der t 
  \\
  &\leq 8 \int_{\Omega_0} 
  |\partial_1(\eta\psi)(s,t)|^2  \, \der s \, \der t 
  + 2 \int_{\Omega_0^I} 
  |\psi(s,t)|^2\, \der s \, \der t 
  \\
  &\leq 16 \int_{\Omega_0} 
  |\partial_1\psi(s,t)|^2  \, \der s \, \der t 
  + (16 \, \|\eta'\|_{L^\infty(\R)}^2 + 2) \int_{\Omega_0^I} 
  |\psi(s,t)|^2\, \der s \, \der t 
  \\
  &\leq 16 \, C
  \int_{\Omega_0} 
  \frac{|\partial_1\psi(s,t)|^2}{f(s,t)}  \, \der s \, \der t 
  + (16 \, \|\eta'\|_{L^\infty(\R)}^2 + 2) \, \|\psi\|_{\Hilbert_I}^2  
  \\
  &\leq 16 \, C
  \ h_1[\psi]  
  + (16 \, \|\eta'\|_{L^\infty(\R)}^2 + 2) \, \|\psi\|_{\Hilbert_I}^2 
  \,.
\end{aligned}
\end{equation}
Here the second estimate follows from the classical Hardy inequality 
$$
  \forall \varphi \in W_0^{1,2}(\R\setminus\{0\})
  \,, \qquad
  \int_{\R} |\varphi'(x) |^2 \, \dd x
  \geq \frac{1}{4}
  \int_{\R} \frac{|\varphi(x)|^2}{x^2} \, \dd x
  \,,
$$
and the last inequality employs~\eqref{lambda} and Lemma~\ref{Lem.lambda}. 
Denoting $K \deff 16 \, \|\eta'\|_{L^\infty(\R)}^2 + 2$
and interpolating between~\eqref{writing0} and~\eqref{writing},
we get
$$
\begin{aligned}
  h_1[\psi] 
  &\geq 
  \delta \, (16\, C)^{-1} 
  \int_{\Omega_0} \frac{|\psi(s,t)|^2}{1+(s-s_0)^2} \, \der s \, \der t
  + [(1-\delta)\, C^{-2}\lambda_0 - \delta \, (16 \,C)^{-1} K] \, \|\psi\|_{\Hilbert_I}^2 
  \\
  &= \frac{\lambda_0}{C(16\,\lambda_0+C K)}
  \int_{\Omega_0} \frac{|\psi(s,t)|^2}{1+(s-s_0)^2} \, \der s \, \der t
  \\
  &\geq 
  \frac{\lambda_0}{C^2(16\,\lambda_0+C K)}
  \int_{\Omega_0} \frac{|\psi(s,t)|^2}{1+(s-s_0)^2} \, f(s,t) \, \der s \, \der t
  \\
  &\geq 
  \frac{\lambda_0}{C^2(16\,\lambda_0+C K)}
  \left(\inf_{s\in\R} \frac{1+s^2}{1+(s-s_0)^2}\right)
  \|\rho^{1/2}\psi\|_\Hilbert^2
  \,.
\end{aligned}
$$
Here the first inequality holds with any $\delta \in \R$
and the equality is due to the choice for which the square bracket vanishes.
The theorem is proved with a constant
$$
  c \geq \frac{\lambda_0}{C^2(16\,\lambda_0+C K)}
  \left(\inf_{s\in\R} \frac{1+s^2}{1+(s-s_0)^2}\right)
  \,,
$$
where the right-hand side depends on the half-width~$a$
and properties of the function~$|\Theta'|$.
\end{proof}

As an immediate consequence of Theorems~\ref{Thm.ess} and~\ref{Thm.Hardy},
we get the following stability result.

\begin{corollary}\label{Corol.Hardy} 
In addition to the hypotheses of Theorem~\ref{Thm.Hardy},
let us assume~\eqref{vanish}. Then 
\begin{equation*}
  \sigma(H) = \sigma_\mathrm{ess}(H) = [E_1,\infty)
  \,.
\end{equation*}
\end{corollary}

However, if~$\Theta'$ does not vanish at the infinity of the strip~$\Omega_0$,
there are situations where the right-hand side of~\eqref{Hardy.global} 
(represented by a positive function vanishing at infinity)
can be replaced by a positive constant (\cf~\cite{KTdA2}). 
In other words, the repulsive effect of twisting 
is so strong that the Hardy inequality turns to a Poincar\'e inequality 
and even the threshold of the essential spectrum grows up.

An obvious application of Theorem~\ref{Thm.Hardy} is the stability
of the spectrum against attractive additive perturbations.
Indeed, in addition to the hypotheses of Theorem~\ref{Thm.Hardy},
let us assume that~$\Theta'$ vanishes at infinity
in the sense of~\eqref{vanish}.
Then, given any bounded function of compact support $V:\Omega_0\to\R$, 
there exists a positive number~$\eps_0$ such that
$\sigma(H+\eps V)=\sigma(H)=[E_1,\infty)$ for every $|\eps| \leq \eps_0$.
Of course, the compact support can be replaced by a fast decay
at infinity comparable to the asymptotic behaviour 
of the Hardy weight~$\rho$.
It is less obvious that the same stability property holds 
against higher-order perturbations, too. 	
As an example, we establish the stability result for 
the purely geometric perturbation of bending.
\begin{theorem}\label{Thm.stability}
Suppose Assumption~\ref{Ass.Ass}. 
Let~$k$ and~$\Theta$ be such that $\Theta' \not= 0$, \eqref{open}~holds 
and the inequality
\begin{equation*}
  |(k \cdot \Theta)(s)| \leq \frac{\eps}{1+s^2}
\end{equation*}
is valid with some non-negative number~$\eps$.
Then  
\begin{equation*}
  H \geq E_1 
\end{equation*}
for all sufficiently small~$\eps$. 
If in addition~\eqref{vanish} holds, then 
\begin{equation*}
  \sigma(H) = \sigma_\mathrm{ess}(H) = [E_1,\infty)
  \,.
\end{equation*}
\end{theorem}
\begin{proof}
The proof is based on the comparison of 
the Jacobian of the full metric~\eqref{Eq metric} 
and the Jacobian without bending~\eqref{metric.unbent}.
Let us keep the notation~$f$ for the former 
and write~$f_0$ for the latter. 
Let us denote $C\deff a\,(2+a\,\|\kappa\|_{L^\infty(\R)})$
and $\varrho(s)\deff1/(1+s^2)$ and assume that $\eps < C^{-1}$. 
Then we have
$$
  1-C\eps\varrho(s)
  \leq \frac{f(s,t)}{f_0(s,t)} \leq
  1+C\eps\varrho(s)
  \,,
$$
for almost every $(s,t) \in \Omega_0$.
In the same manner, let us keep the notation~$h$ and~$\Hilbert$
respectively for the form~\eqref{form} and the Hilbert space~\eqref{Hilbert}
corresponding to~$f$ and let us write~$h_0$ and~$\Hilbert_0$
for the analogous quantities corresponding to~$f_0$.
Let $\psi \in C_0^\infty(\Omega_0)$,
a dense subspace of both $\Dom(h)$ and $\Dom(h_0)$.
Using the estimates on the Jacobians above, we get
$$
\begin{aligned}
  h[\psi] - E_1 \|\psi\|_\Hilbert^2
  \geq \ & 
  \frac{1}{1-C\eps}
  \int_{\Omega_0} \frac{|\partial_1\psi(s,t)|^2}{f_0(s,t)} \, \der s \, \der t
  \\
  & + \int_{\Omega_0} (1-C\eps\varrho(s)) 
  \left(|\partial_2\psi(s,t)|^2 - E_1 |\psi(s,t)|^2 \right)
  f_0(s,t) 
  \, \der s \, \der t
  \\
  & + E_1 \int_{\Omega_0} 2\,C\eps\varrho(s) \, |\psi(s,t)|^2\, f_0(s,t) 
  \, \der s \, \der t
  \,.
\end{aligned}
$$   
Since the integrand on the second line is non-negative 
due to~\eqref{lambda} and Lemma~\ref{Lem.lambda},
we get the estimate
$$
  h[\psi] - E_1 \|\psi\|_\Hilbert^2
  \geq (1-C\eps) \left(h_0[\psi] - E_1 \|\psi\|_{\Hilbert_0}^2\right)
  + E_1 \int_{\Omega_0} 2\,C\eps\varrho(s) \, |\psi(s,t)|^2\, f_0(s,t) 
  \, \der s \, \der t 
  \,.
$$   
Applying the Hardy inequality of Theorem~\ref{Thm.Hardy}, 
we conclude with
$$
  h[\psi] - E_1 \|\psi\|_\Hilbert^2 \geq
  \int_{\Omega_0} (c-E_1 2\,C\eps) \, \varrho(s) \, |\psi(s,t)|^2\, f_0(s,t) 
  \, \der s \, \der t 
  \,.
$$
If $\eps \leq c/(E_1 2\,C)$, it follows that $H \geq E_1$.
In fact, we have established the Hardy inequality 
$$
  H - E_1 \geq (c-E_1 2\,C\eps) \, \rho \, \frac{f_0}{f}
$$
if the strict inequality $\eps < c/(E_1 2\,C)$ holds.
Assuming now~\eqref{vanish}, the fact that all energies $[E_1,\infty)$
belong to the spectrum of~$H$ follows by Theorem~\ref{Thm.ess}.  
\end{proof}
\begin{open} 
Is the smallness condition~\eqref{open}
necessary for the existence of the Hardy inequality?
\end{open}
\begin{open}
It follows from Corollary~\ref{Corol.Hardy} 
that~$H$ possesses no discrete eigenvalues. 
Is it true that, under the hypotheses of Corollary~\ref{Corol.Hardy}, 
there are no (embedded) eigenvalues inside the interval $[E_1,\infty)$ either?
On the other hand, it has been recently observed 
in~\cite{Bruneau-Miranda-Popoff_2018,Bruneau-Miranda-Parra-Popoff} 
that a local twist of a solid waveguide 
leads to the appearance of resonances around 
the thresholds given by the eigenvalues of the cross-section.  
Does this phenomenon occurs in the twisted strips as well?
\end{open}
\begin{open} 
Solid tubes with asymptotically diverging twisting
represent a new class of models which lead to previously 
unobserved phenomena like the annihilation of the essential spectrum~\cite{K11}
and establishing a non-standard Weyl's law for the accumulation 
of eigenvalues at infinity remains open
(a first step in this direction has been recently taken 
in \cite{Barseghyan-Khrabustovskyi_2018} by establishing 
a Berezin-type upper bound for the eigenvalue moments).
The case of twisted strips with $|\Theta'(s)|\to\infty$ as $|s| \to \infty$
is rather different for some essential spectrum is always present, 
but related questions about the accumulation of discrete eigenvalues
remain open, too (\cf~\cite{KTdA2}).
\end{open}
%

%------------------------------%
\section{Thin strips}\label{Sec.thin}
%------------------------------%
%
In this last section, we consider simultaneously bent and twisted strips
in the limit when the half-width~$a$ tends to zero.
Since we consider Dirichlet boundary conditions,
it is easily seen that $\inf\sigma(H) \to \infty$ as $a \to 0$.
However, a non-trivial limit is obtained 
for the ``renormalised'' operator $H-E_1$. 
Roughly, we shall establish the limit
\begin{equation}\label{formal}
  H-E_1 
  \ \xrightarrow[a \to 0]{} \ 
  H_\mathrm{eff} 
  \deff -\frac{\der^2}{\der s^2} + V_\mathrm{eff}(s)
  \,,
\end{equation}
where~$H_\mathrm{eff}$ is an operator in~$\sii(\R)$  
and the geometric potential~$V_\mathrm{eff}$ provides a valuable insight
into the opposite effects of bending and twisting:
$$
  V_\mathrm{eff}  
  \deff - \frac{1}{4} \, (k\cdot\Theta)^2 + \frac{1}{2} \, |\Theta'|^2
  \,. 
$$
That is, the geodesic curvature of~$\Gamma$ 
as a curve on~$\Omega$ realises 
an attractive part of the potential, 
while the Gauss curvature of the ambient surface~$\Omega$ 
acts as a repulsive interaction.
Since the operators~$H$ and~$H_\mathrm{eff}$ are unbounded
and, moreover, they act in different Hilbert spaces, 
it is necessary to properly interpret the formal limit~\eqref{formal}.

We start by transferring~$H$ into a unitarily operator 
in the $a$-independent Hilbert space $\sii(\Pi)$ 
with $\Pi \deff \R \times (-1,1)$. 
This is achieved by the unitary transform 
$
  U:\Hilbert \to \sii(\Pi)
$
defined by
$$
  (U\psi)(s,u) \deff \sqrt{a \, f(s,a u)} \, \psi(s,a u)
  \,.
$$
We shall write $f_a(s,u) \deff f(s,a u)$.
The unitarily equivalent operator $\hat{H} \deff UHU^{-1}$ in $\sii(\Pi)$
is the operator associated with the quadratic form
$\hat{h}[\phi] \deff h[U^{-1}\phi]$, $\Dom(\hat{h}) \deff U \Dom(h)$.  
It will be convenient to strengthen Assumption~\ref{Ass.Ass}.

\begin{center}
\fbox{%
\begin{minipage}{0.98\textwidth}
\begin{Assumption}\label{Ass.thin}
Let $\Gamma \in C^{2,1}(\R;\R^{n+1})$ 
and $\Theta \in C^{1,1}(\R;\R^{n})$.
Suppose~\eqref{Eq sum of theta 2},
$
  a \, \| k\cdot\Theta \|_{L^\infty(\R)} < 1  
$
and  
\begin{equation*}%\label{Ass}
  k\cdot\Theta, \, (k\cdot\Theta)', \, |\Theta'|, \, |\Theta''| \in L^\infty(\R)
  \,.
\end{equation*}
\end{Assumption}
\end{minipage}%
}
\end{center}

The inequality of the assumption does not need to
be explicitly assumed, for it will be always satisfied 
for all sufficiently small~$a$.
Again, neither the curvature~$\kappa$ nor its derivative~$\kappa'$
are assumed to be (globally) bounded, \cf~Remark~\ref{Rem.spiral}. 

Since Assumption~\ref{Ass.thin} particularly involves~\eqref{Ass2},
we have the global bounds~\eqref{global} to the Jacobian~$f$,
and consequently $\Dom(h) = W_0^{1,2}(\Omega_0)$.
Given any $\phi \in C_0^\infty(\Pi)$,
the H\"older continuity hypotheses of Assumption~\ref{Ass.thin} 
ensure that $U^{-1} \phi \in \Dom(h)$.
A straightforward computation yields
$$
  \hat{h} = \hat{h}_1 + \hat{h}_2
$$
with
$$
\begin{aligned}
  \hat{h}_1[\phi] 
  &\deff 
  \int_{\Pi} \frac{|\partial_1\phi|^2}{f_a^2} \, \der s \, \der u
  + \frac{1}{4} \int_{\Pi} 
  \frac{(\partial_1 f_a)^2}{f_a^4} \, 
  |\phi|^2 \, \der s \, \der u
  - \Re \int_{\Pi} 
  \frac{\partial_1 f_a}{f_a^2} \, 
  \overline{\phi} \, \partial_1\phi \, \der s \, \der u
  \,,
  \\
  \hat{h}_2[\phi] 
  &\deff 
  \frac{1}{a^2} 
  \int_{\Pi} |\partial_2\phi|^2 \, \der s \, \der u
  + \frac{1}{4 a^2} \int_{\Pi} 
  \frac{(\partial_2 f_a)^2}{f_a^2} \, 
  |\phi|^2 \, \der s \, \der u
  - \frac{1}{a^2} \, \Re \int_{\Pi} 
  \frac{\partial_2 f_a}{f_a} \, 
  \overline{\phi} \, \partial_2\phi \, \der s \, \der u 
  \,,
\end{aligned}
$$
where we suppress the arguments $(s,u)$ 
of the integrated functions for brevity.
Integrating by parts in the second form, 
we further get
$$
  \hat{h}_2[\phi] 
  = \frac{1}{a^2} 
  \int_{\Pi} |\partial_2\phi|^2 \, \der s \, \der u
  +  \int_{\Pi} V_a \, |\phi|^2 \, \der s \, \der u
$$
with
$$
  V_a \deff -\frac{1}{4 a^2} \frac{(\partial_2 f_a)^2}{f_a^2}
  + \frac{1}{2 a^2} \frac{\partial_2^2 f_a}{f_a}
  \,.
$$
Using~\eqref{global} together with
the uniform boundedness hypotheses of Assumption~\ref{Ass.thin},
is is easy to verify that 
$$
  \Dom(\hat{h}) = W_0^{1,2}(\Pi)
  \,.
$$

Using Assumption~\ref{Ass.thin}, one has the estimates
\begin{equation}\label{uniform}
  \|f_a - 1\|_{L^\infty(\Pi)} \leq C \, a
  \,, \qquad
  \|\partial_1 f_a\|_{L^\infty(\Pi)} \leq C \, a 
  \,, \qquad
  \|V_a - V_\mathrm{eff}\|_{L^\infty(\Pi)} \leq C \, a
  \,,
\end{equation}
where~$C$ is a constant depending on the supremum norms of
$k\cdot\Theta$, $(k\cdot\Theta)'$, $|\Theta'|$ and $|\Theta''|$.
It is therefore expected that~$\hat{H}$
will be, in the limit as $a \to 0$, well approximated 
by the operator~$\hat{H}_0$ associated with the form   
$$
\begin{aligned}
  \hat{h}_0[\phi] 
  &\deff 
  \int_{\Pi} |\partial_1\phi|^2 \, \der s \, \der u
  + \frac{1}{a^2} 
  \int_{\Pi} |\partial_2\phi|^2 \, \der s \, \der u
  + \int_{\Pi} V_\mathrm{eff} \, |\phi|^2 \, \der s \, \der u
  \,,
  \\
  \Dom(\hat{h}_0) &\deff W_0^{1,2}(\Pi)
  \,,
\end{aligned}
$$
where we keep the same notation~$V_\mathrm{eff}$ 
for the function $V_\mathrm{eff} \otimes 1$ on $\R \times (-1,1)$.
Here we establish the closeness of 
the operators~$\hat{H}$ and~$\hat{H}_0$
in a norm resolvent sense. 
To formulate the result, we note that the bound 
\begin{equation}\label{z0}
  V_\mathrm{eff}(s) \geq 
  - \frac{1}{4} \, \|k\cdot\Theta\|_{L^\infty(\R)}^2 
  \deffin z_0
\end{equation}
and the Poincar\'e inequality~\eqref{Poincare}
imply that $\hat{H}_0 - E_1 \geq z_0$.
Hence, any $z < z_0$ certainly belongs to the resolvent set of~$\hat{H}_0$.

\begin{theorem}\label{Thm.thin}
Suppose Assumption~\ref{Ass.thin}.
For every $z < z_0$ there exist positive numbers~$a_0$ and~$C$ 
such that, for all $a \leq a_0$, $z \in \rho(\hat{H})$ and
\begin{equation}\label{nrs}
  \|(\hat{H}-E_1-z)^{-1} - (\hat{H}_0-E_1-z)^{-1}\|_{\sii(\Pi)\to\sii(\Pi)} 
  \leq C \, a
  \,.
\end{equation}
\end{theorem}
\begin{proof}
Let us write $\|\cdot\|$ and $(\cdot,\cdot)$ 
for the norm and inner product of~$\sii(\Pi)$, respectively.
Given any $F_0 \in \sii(\Pi)$, 
let $\phi_0 \in \Dom(\hat{H}_0)$ 
be the (unique) solution of the resolvent equation 
$(\hat{H}_0 - E_1 - z)\phi_0=F_0$.
Using the Schwarz inequality and~\eqref{Poincare}, 
one has the estimates
$$
  \|\partial_1\phi_0\|^2 + (z_0 - z) \|\phi_0\|^2
  \leq \hat{h}_0[\phi_0] - E_1 \|\phi_0\|^2 - z \|\phi_0\|^2
  = (\phi_0,F_0)
  \leq \|\phi_0\| \|F_0\|
  \,.
$$
Consequently, 
\begin{equation}\label{uniform1}
  \|\phi_0\| \leq C \, \|F_0\| 
  \qquad \mbox{and} \qquad
  \|\partial_1\phi_0\| \leq C \, \|F_0\| 
  \,,
\end{equation}
where~$C$ is a positive constant depending exclusively on~$z_0-z$. 
From now on, we denote by~$C$ a generic constant
(possibly depending on~$z$ and the supremum norms of
$k\cdot\Theta$, $(k\cdot\Theta)'$, $|\Theta'|$ and $|\Theta''|$),
which might change from line to line.
 
For every $\phi \in W_0^{1,2}(\Pi)$ and $\delta \in (0,1)$, 
we have
$$
\begin{aligned}
  \hat{h}_1[\phi] 
  &\geq
  \delta \int_{\Pi} \frac{|\partial_1\phi|^2}{f_a^2} \, \der s \, \der u
  - \frac{1}{4} \frac{\delta}{1-\delta} \int_{\Pi} 
  \frac{(\partial_1 f_a)^2}{f_a^4} \, 
  |\phi|^2 \, \der s \, \der u
  \\
  &\geq \delta \, C^{-1} \, \|\partial_1\phi\|^2 
  - \frac{1}{4} \frac{\delta}{1-\delta} \, C \, a^2 \, \|\phi\|^2  
  \,,
\end{aligned}
$$
where the second inequality is due to~\eqref{uniform}.
At the same time, using additionally~\eqref{Poincare}, we have
$$
\begin{aligned}
  \hat{h}_2[\phi] 
  &=
  \frac{1}{a^2} \, \int_{\Pi} |\partial_2\phi|^2 \, \der s \, \der u
  + \int_{\Pi} V_\mathrm{eff} \, |\phi|^2 \, \der s \, \der u
  + \int_{\Pi} (V_a-V_\mathrm{eff}) \, |\phi|^2 \, \der s \, \der u
  \\
  &\geq E_1 \, \|\phi\|^2  
  + z_0 \, \|\phi\|^2   
  - C \, a \, \|\phi\|^2
  \,.
\end{aligned}
$$
Consequently, $\hat{H}-E_1-z \geq z_0-z-Ca$
(with a possibly different constant~$C$).
Since $z_0-z$ is positive, it follows that there exists a positive number~$a_0$ such that,
for all $a \leq a_0$, the number~$z$ belongs to the resolvent set of~$\hat{H}$. 
Given any $F \in \sii(\Pi)$, 
let $\phi \in \Dom(\hat{H})$  
be the (unique) solution of the resolvent equation
$(\hat{H} - E_1 - z)\phi=F$.
The above estimates imply 
\begin{equation}\label{uniform2}
  \|\phi\| \leq C \, \|F\| 
  \qquad \mbox{and} \qquad
  \|\partial_1\phi\| \leq C \, \|F\| 
  \,,
\end{equation}

Now we write
$$
\begin{aligned}
  \big(F,[(\hat{H}-E_1-z)^{-1} - (\hat{H}_0-E_1-z)^{-1}]F_0\big) 
  &= \big(\phi,(\hat{H}_0-E_1-z)\phi_0\big) 
  - \big((\hat{H}-E_1-z)\phi,\phi_0\big) 
  \\
  &= \hat{h}_0(\phi,\phi_0) - \hat{h}(\phi,\phi_0)
  \,,
\end{aligned}
$$
where the last equality follows from the fact that
the operators~$\hat{H}$ and~$\hat{H}_0$ have the same form domains.
Using the structure of the forms~$\hat{h}$ and~$\hat{h}_0$, 
we estimate the difference of the sesqulinear forms as follows:
$$
\begin{aligned}
  |\hat{h}_0(\phi,\phi_0) - \hat{h}(\phi,\phi_0)|
  \leq \ & 
  \|1-f_a^{-2}\|_{L^\infty(\Pi)} \|\partial_1\phi\| \|\partial_1\phi_0\|
  + \frac{1}{4} \|f_a^{-4}(\partial_1 f_a)^2\|_{L^\infty(\Pi)}
  \|\phi\| \|\phi_0\|
  \\
  &
  + \frac{1}{2} \|f_a^{-2}(\partial_1 f_a)\|_{L^\infty(\Pi)}
  (\|\partial_1\phi\| \|\phi_0\| + \|\phi\| \|\partial_1\phi_0\|)
  \\
  &
  + \|V_a-V_\mathrm{eff}\|_{L^\infty(\Pi)} \|\phi\| \|\phi_0\|
  \,.
\end{aligned}
$$
Using~\eqref{uniform} and~\eqref{uniform1}--\eqref{uniform2},
it follows that 
$$
  \left|
  \big(F,[(\hat{H}-E_1-z)^{-1} - (\hat{H}_0-E_1-z)^{-1}]F_0\big) 
  \right|
  \leq C \, a \, \|F\| \|F_0\|
  \,.
$$
Dividing by $\|F\| \|F_0\|$ and taking the supremum over all $F,F_0 \in \sii(\Pi)$,
we arrive at~\eqref{nrs}.  
\end{proof}

As a particular consequence of Theorem~\ref{Thm.thin},
we get a certain convergence of the spectrum of~$\hat{H}$
to the spectrum of the one-dimensional operator~$H_\mathrm{eff}$. 
Indeed, by a separation of variables,
the spectrum of~$\hat{H}_0$ decouples as follows:
$$
  \sigma(\hat{H}_0-E_1) 
  = \bigcup_{j=1}^\infty [\sigma(H	_\mathrm{eff}) + E_j-E_1]
  \,,
$$
where $E_j \deff \left(\frac{j\pi}{2a}\right)^2$ 
are the eigenvalues of the Dirichlet Laplacian in $\sii((-a,a))$.
It follows that the spectrum of $\hat{H}_0-E_1$
converges to the spectrum of~$H_\mathrm{eff}$ 
in the sense that, given any positive number~$L$,
there is another positive number~$a_L$ such that,
for all $a \leq a_L$, one has
$$
  \sigma(\hat{H}_0-E_1) \cap (-\infty,L)
  = \sigma(H_\mathrm{eff}) \cap (-\infty,L)
  \,.
$$
Theorem~\ref{Thm.thin} particularly implies 
that for any discrete eigenvalue of~$H_\mathrm{eff}$,
there is a discrete eigenvalue of~$\hat{H}-E_1$ 
(and therefore of~$H-E_1$) which converges to the former as $a \to 0$.
A convergence in norm of corresponding spectral projections also follows.

What is more, the spectral convergence follows as a consequence
of a norm resolvent convergence again.
To see it, we define the orthogonal projection
$$
  (P \psi)(s,u) \deff 
  \hat{\chi}_1(u) \int_{-1}^{1} \psi(s,\eta) \, \hat{\chi}_1(\eta) \, \der \eta
  \,,
$$
where $\hat{\chi}_1(u) \deff \sqrt{a} \, \chi_1(au)$
with~$\chi_1$ being the first eigenfunction 
of the Dirichlet Laplacian in $\sii((-a,a))$, see~\eqref{chi1}.
The closed subspace $P\sii(\Pi)$ obviously consists of functions
of the form $(s,u) \mapsto \varphi(s)\chi_1(u)$, where $\varphi \in \sii(\R)$.
The mapping $\pi:\sii(\R) \to P\sii(\Pi)$ defined by
$(\pi\varphi)(s,u) \deff \varphi(s)\chi_1(u)$ is an isometric isomorphism.
In this way, we may canonically identify any operator~$T$ acting in $\sii(\R)$ 
with the operator $\pi T \pi^{-1}$ in $P\sii(\Pi) \subset \sii(\Pi)$.
In particular, we use the same symbol~$H_\mathrm{eff}$ 
for the corresponding operator in $P\sii(\Pi)$,
and similarly for its resolvent.  
With these preliminaries, the desired result reads as follows.

\begin{proposition}\label{Prop.thin}
Suppose Assumption~\ref{Ass.thin}.
For every $z < z_0$ one has
\begin{equation}\label{nrs0}
  \|
  (\hat{H}_0-E_1-z)^{-1} - (H_\mathrm{eff}-z)^{-1} \oplus 0
  \|_{\sii(\Pi)\to\sii(\Pi)} 
  \leq \frac{3}{2\pi} \, a
  \,,
\end{equation}
where~$0$ denotes the zero operator on $\sii(\Pi) \ominus P\sii(\Pi)$.
\end{proposition}
\begin{proof}
Defining $P^\bot \deff I-P$, we have the identity
$$
\begin{aligned}
  (\hat{H}_0-E_1-z)^{-1}
  = \ &P (\hat{H}_0-E_1-z)^{-1} P + P^\bot (\hat{H}_0-E_1-z)^{-1} P^\bot
  \\
  & + P (\hat{H}_0-E_1-z)^{-1} P^\bot
  + P^\bot (\hat{H}_0-E_1-z)^{-1} P
  \\
  = \ &
  (H_\mathrm{eff}-z)^{-1} P 
  + P^\bot (\hat{H}_0-E_1-z)^{-1} P^\bot
  \,.
\end{aligned}
$$
Given any $F \in \sii(\Pi)$,
let~$\psi$ be the (unique) solution of the resolvent equation
$(\hat{H}_0-E_1-z)\psi=P^\bot F$. 
That is, $\psi \in \Dom(\hat{H}_0) \subset \Dom(\hat{h}_0)$
and, for every $\phi \in \Dom(\hat{h}_0)$,
$$
  \hat{h}_0(\phi,\psi) - (E_1+z) (\phi,\psi)
  = (\phi,P^\bot F) \leq \|\phi\| \|P^\bot F\|
  \,. 
$$ 
Choosing $\phi \deff P^\bot\psi$ and using~\eqref{z0}   
together with the facts that 
$
  (\partial_1 P^\bot\psi,\partial_1\psi)
  = \|\partial_1 P^\bot\psi\|^2 \geq 0
$ 
and 
$
  (\partial_2 P^\bot\psi,\partial_2\psi)
  = \|\partial_2 P^\bot\psi\|^2 
  \geq E_2 \|P^\bot\psi\|^2 
$,
we therefore get 
$$
  \|P^\bot\psi\| \leq \frac{1}{E_2-E_1} \, \|P^\bot F\|
  = \frac{a^2}{3\pi^2} \, \|P^\bot F\|
  \,.
$$
Consequently,
$$
  \big|\big(F,P^\bot (\hat{H}_0-E_1-z)^{-1} P^\bot F\big)\big|
  \leq \frac{a^2}{3\pi^2} \, \|P^\bot F\|^2
  \leq \frac{a^2}{3\pi^2} \, \|F\|^2
  \,.
$$
In view of the resolvent identity above,
this proves the desired claim.
\end{proof}

Combining Theorem~\ref{Thm.thin} and Proposition~\ref{Prop.thin}
and recalling the unitary equivalence of~$H$ and~$\hat{H}$,
we have just justified the formal statement~\eqref{formal} 
in a rigorous way of a norm resolvent convergence.

\begin{corollary}\label{Corol.thin} 
Suppose Assumption~\ref{Ass.thin}.
For every $z < z_0$ there exist positive numbers~$a_0$ and~$C$ 
such that, for all $a \leq a_0$, $z \in \rho(\hat{H})$ and
\begin{equation*}%\label{nrs}
  \|(\hat{H}-E_1-z)^{-1} 
  - (H_\mathrm{eff}-z)^{-1} \oplus 0\|_{\sii(\Pi)\to\sii(\Pi)} 
  \leq C \, a
  \,.
\end{equation*}
\end{corollary}

We note that this result has been previously established 
by Verri~\cite{Verri_2019} in the special setting of purely twisted strips.
In fact, in recent years there has been an exponential growth of interest 
in effective models for thin waveguides 
under various geometric and analytic deformations,
see \cite{Ferreira-Mascarenhas-Piatnitski_2015,
Jimbo-Kurata_2016,
Keller-Teufel_2016,Lampart-Teufel_2017,Mehats-Raymond_2017,
Oliveira-Verri_2017a,Oliveira-Verri_2017b,
Mamani-Verri_2018a,Mamani-Verri_2018b,Yachimura_2018,
deOliveira-Hartmann-Verri_2019,
Verri_2019,Oliveira-Rossini} 
%to quote just the most recent 
and further references therein.
We refer to~\cite{KRRS2} for a unifying approach to this type of problems.

\begin{open} 
Following~\cite{Mamani-Verri_2018b},
locate the band gaps in thin periodically twisted and bent strips.
\end{open}

\appendix

%------------------------------------------%
\section{Relatively parallel frame}\label{Sec.app}
%------------------------------------------%
%
In this appendix we establish a purely geometric fact
about the existence of a \emph{relatively parallel adapted frame} 
for any curve $\Gamma:I\to\R^{n+1}$,
where $n \geq 1$ and $I \subset \R$ is an arbitrary open interval
(bounded or unbounded).  
Our primary motivation is to generalise 
the approach of Bishop~\cite{Bishop_1975} for $n+1=3$  
to any space dimension.
Secondarily,
and contrary to Bishop who assumes that the curve~$\Gamma$ is of class~$C^2$,
we proceed under the minimal hypothesis 
\begin{equation}\label{minimal}
  \Gamma \in C^{1,1}(I;\R^{n+1}) 
  \,,
\end{equation}
which is natural for applications 
(like, for instance, in the theory of quantum waveguides
considered in this paper).
For three-dimensional curves, the latter generalisation
has been already performed in~\cite{KSed}.

Without loss of generality, we assume that~$\Gamma$
is parameterised by its arc-length, 
\ie~$|\Gamma'(s)|=1$ for all $s \in \R$.
Then $T\deff\Gamma'$ defines a unit tangent vector field along~$\Gamma$,
which is locally Lipschitz continuous 
and as such it is differentiable almost everywhere in~$I$.
The non-negative number $\kappa \deff |\Gamma''|$ 
is called the \emph{curvature} of~$\Gamma$.
It is worth noticing that the curvature~$\kappa$ is not assumed 
to be (globally) bounded by~\eqref{minimal}. 
In particular, $\Gamma$~is allowed to be 
a spiral with $\kappa(s) \to \infty$ as $s \to \pm\infty$.

An \emph{adapted frame} of~$\Gamma$ 
is the $(n+1)$-tuple $(T,N_1,\dots,N_n)$ 
of orthonormal vector fields along~$\Gamma$,
which are differentiable almost everywhere in~$I$.
We say that a normal vector field~$N$ along~$\Gamma$
is \emph{relatively parallel} if~$N$  
is differentiable almost everywhere in~$I$
and the derivative~$N'$ is tangential 
(\ie~there exists a locally bounded function $k:I\to \R$ such that $N'=k T$).
Notice that any relatively parallel vector field~$N$ along~$\Gamma$
has a constant length (indeed, ${N^2}'=2N \cdot N' = 0$).
By a \emph{relatively parallel adapted frame} of~$\Gamma$
we then mean an adapted frame $(T,N_1,\dots,N_n)$
such that the normal vector fields 
$N_1,\dots,N_n$ are relatively parallel.  
Consequently, the relatively parallel adapted frame satisfies the equation
\begin{equation}\label{frame.bis}
\begin{pmatrix}
    T \\
    N_1 \\
    \vdots \\
    N_n
\end{pmatrix}'
=
\begin{pmatrix}
    0 & k_1 & \dots  & k_n \\
   -k_1 &  0 & \dots  & 0 \\
    \vdots & \vdots & \ddots & \vdots \\
    -k_n  & 0 & \dots  & 0
\end{pmatrix}
\begin{pmatrix}
T \\
    N_1 \\
    \vdots \\
    N_n
\end{pmatrix}
,
\end{equation}
where $k_1,\dots,k_n \in L_\mathrm{loc}^\infty(I)$. 
Necessarily, $k_1^2+\dots+k_n^2=\kappa^2$.

\begin{example}[Frenet frame]\label{Ex.Frame}
If $\Gamma \in C^{1,1}(I;\R^2)$, 
then the Frenet frame $(T,N)$ with $N\deff(-\Gamma_2',\Gamma_1')$
is a relatively parallel adapted frame of~$\Gamma$.
Indeed, one has the Frenet-Serret formulae
$$
  \begin{pmatrix}
    T \\ N
  \end{pmatrix}'
  =
  \begin{pmatrix}
    0 & k \\
    -k & 0
  \end{pmatrix}
  \begin{pmatrix}
    T \\ N
  \end{pmatrix}
  ,
$$
where the \emph{signed curvature} 
$k \deff -\Gamma_1''\Gamma_2'+\Gamma_1'\Gamma_2''$ satisfies $|k|=\kappa$.

Let $\Gamma \in C^{2,1}(I;\R^3)$ and assume that $\kappa>0$,
so that the \emph{principal normal} $M_1\deff\Gamma''/|\Gamma''|$ 
is well defined.	
Defining the \emph{binormal} $M_2 \deff T \times M_1$,
it is customary to consider the Frenet frame $(T,M_1,M_2)$.
The Frenet-Serret equations read
$$
  \begin{pmatrix}
    T \\ M_1 \\ M_2
  \end{pmatrix}'
  =
  \begin{pmatrix}
    0 & \kappa & 0 \\
    -\kappa & 0 & \tau \\
   0 & -\tau & 0 
  \end{pmatrix}
  \begin{pmatrix}
    T \\ M_1 \\ M_2
  \end{pmatrix}
  ,
$$
where $\tau\deff\det(\Gamma',\Gamma'',\Gamma''')/\kappa^2$ 
is the \emph{torsion} of~$\Gamma$.
Consequently, the Frenet frame is a relatively parallel adapted frame
if, and only, if $\tau=0$, \ie, $\Gamma$~lies in a plane.  

In general, let $\Gamma \in C^{n,1}(I;\R^{n+1})$ with $n \geq 1$
and assume that the vector fields
$\Gamma',\Gamma'',\dots,\Gamma^{(n)}$ are linearly independent.
By applying the Gram-Schmidt orthogonalisation process 
to $\Gamma',\Gamma'',\dots,\Gamma^{(n)}$,
it is easily seen (see, \eg, \cite[Prop.~1.2.2]{Kli})
that there exists a Frenet frame $(T,M_1,\dots,M_n)$ 
satisfying the equations
$$
  \begin{pmatrix}
    T \\ M_1 \\ \vdots \\ M_n
  \end{pmatrix}'
  =
  \begin{pmatrix}
   0                & \kappa_1 &               & \textrm{\Large 0}\\
   -\kappa_1        & \ddots   & \ddots        &                 \\
                    & \ddots   & \ddots        & \kappa_{n}    \\
   \textrm{\Large 0} &          & -\kappa_{n} & 0
   \end{pmatrix}
  \begin{pmatrix}
    T \\ M_1 \\ \vdots \\ M_n
  \end{pmatrix}
$$
with some locally bounded functions $\kappa_1,\dots,\kappa_n$ 
actually defined by these formulae.
Again, the Frenet frame is a relatively parallel adapted frame
if, and only, if all the higher curvatures $\kappa_2,\dots,\kappa_n$ equal to zero, 
\ie, $\Gamma$~lies in a plane. 
We refer to~\cite{ChDFK} for a construction of the Frenet frame 
under weaker hypotheses about~$\Gamma$. 
\end{example} 

The defect of working with the Frenet frame is that it requires
at least the regularity $\Gamma \in C^{n,1}(I;\R^{n+1})$.
Moreover, the non-degeneracy condition that the vector fields
$\Gamma',\Gamma'',\dots,\Gamma^{(n)}$ are linearly independent
is indeed necessary in general (\cf~\cite[Chapt.~1]{Spivak1}).
Fortunately, its alternative given by the relatively parallel adapted frame
always exists, and moreover the minimal hypothesis~\eqref{minimal} is enough.
\begin{theorem}\label{Thm.app}
Suppose~\eqref{minimal}.
Let $(T(s_0),N_1^0,\dots,N_n^0)$ be an orthonormal basis 
of the tangent space $\mathbb{T}_{\Gamma(s_0)}\R^{n+1}$ for some $s_0 \in I$.
Then there exists a unique relatively parallel adapted frame 
$(T,N_1,\dots,N_n)$ of~$\Gamma$ such that 
$N_j(s_0) = N_j^0$ for every $j \in \{1,\dots,n\}$.
\end{theorem}
\begin{proof}
We divide the proof into several steps.

\smallskip
\noindent
\emph{Uniqueness.}
Assume that there exists another relatively parallel adapted frame
$(T,M_1,\dots,M_n)$ of~$\Gamma$ such that 
$N_j(s_0)=M_j(s_0)$ for every $j \in \{1,\dots,n\}$.
Then $(T,N_1-M_1,\dots,N_n-M_n)$ is also 
a relatively parallel adapted frame of~$\Gamma$.
The uniqueness follows by the general fact
that the length of any relatively parallel vector field is preserved
and by the hypothesis that the length of the difference $N_i-M_i$ at~$s_0$ is zero.

\smallskip
\noindent
\emph{Local existence of an adapted frame.}
Let $s_0$ be an arbitrary point of~$I$
and let us set $d_0 \deff \dist(s_0,\partial I)$.
From the identity $1=|T|^2=T_1^2 + \dots + T_{n+1}^2$ on $I$,
it follows that there exists at least one index 
$j \in \{1, \dots,n+1\}$ such that
\begin{equation*}
  T_j(s_0)^2  \geq \frac{1}{n+1} 
  \,.
\end{equation*}
Without loss of generality, we can assume that $j = n+1$. 
Since~$T$ is continuous, 
there must exist some $\eps \in (0,d_0]$ such that $|T| > 0$ 
on $(s_0 - \varepsilon, s_0 + \varepsilon) \subset I$. 
More specifically, using the identity
$$
  T(s) - T(s_0) = \int_{s_0}^s T'(\xi) \, \der\xi
$$
together with the fact that $\|T'\|=\kappa$
and that the curvature is locally bounded,
we may take
\begin{equation}\label{epsilon}
  \eps \deff \min\left\{
  \frac{d_0}{2},
  \frac{1}{\sqrt{n+1}\, \|\kappa\|_{L^\infty((s_0-d_0/2,s_0+d_0/2))}}
  \right\}
\end{equation}
(with the convention that the minimum equals $d_0/2$
if the supremum norm of the curvature equals zero).
Defining
\begin{align*}
  \tilde{M}_1 &\deff 
  \frac{1}{\sqrt{T_{n+1}^2 + T_1^2}} \,
  (T_{n+1}, 0, 0, \dots, 0, 0, - T_1) 
  \,,
  \\
  \tilde{M}_2 &\deff 
  \frac{1}{\sqrt{T_{n+1}^2 + T_2^ 2}} \,
  (0, T_{n+1}, 0, \dots, 0, 0, -T_2) 
  \,,
  \\
  \vdots 
  \\
  \tilde{M}_n & \deff 
  \frac{1}{\sqrt{T_{n+1}^2 + T_n^2}} \,
  (0, 0, 0,  \dots, 0, T_{n+1}, -T_n)
  \,,
\end{align*}
it is clear that, on $(s_0 - \varepsilon, s_0 + \varepsilon)$,
these vector fields are linearly independent,
orthogonal to the tangent vector~$T$ and of unit length. 
However, they do not need to be mutually orthogonal.
The desired adapted frame $(T, M_1, \dots, M_n)$ 
on $(s_0 - \varepsilon, s_0 + \varepsilon)$
is obtained by applying the Gram-Schmidt orthogonalisation procedure to 
$(T, \tilde{M}_1, \dots, \tilde{M}_n)$. 

\smallskip
\noindent
\emph{Local existence of the relatively parallel adapted frame.}
Let~$s_0$ be the given point of~$I$ from the statement of the theorem. 
By the preceding construction, we have an adapted frame $(T, M_1, \dots, M_n)$ 
of $\Gamma \upharpoonright (s_0 - \varepsilon, s_0 + \varepsilon)$.
It satisfies the equation
\begin{equation*} 
\begin{pmatrix}
    T \\
    M_1 \\
    \vdots \\
    M_n
\end{pmatrix}'
=
\begin{pmatrix}
    a_{00} & a_{01} & \dots  & a_{0n} \\
    a_{10} &  a_{11} & \dots  & a_{1n} \\
    \vdots & \vdots & \ddots & \vdots \\
    a_{n0} & a_{n1} & \dots  & a_{nn}
\end{pmatrix}
\begin{pmatrix}
    T \\
    M_1 \\
    \vdots \\
    M_n
\end{pmatrix}
,
\end{equation*}
where $a_{\mu\nu} \in L^\infty((s_0-\eps,s_0+\eps))$ 
are such that $a_{\mu\nu} = -a_{\nu\mu}$ with $\mu,\nu \in \{0,\dots,n\}$.
Hence the matrix-valued function 
$A \deff (a_{\mu\nu})_{\mu,\nu=0}^n$ is skew-symmetric. 
It will be convenient to express it as follows:
$$
  A = 
  \begin{pmatrix}
    0 & a^T \\
    -a & \tilde{A}
  \end{pmatrix}
  \,,
$$
where $\tilde{A} \deff (a_{jk})_{j,k=1}^n$ and $a^T \deff (a_{01},\dots,a_{0n})$.
Let us consider a generic orthogonal matrix-valued function 
$\tilde{R} \deff (r_{jk})_{j,k=1}^n$ 
satisfying $r_{jk} \in C^{0,1}([s_0-\eps,s_0+\eps])$ and define
$$
  R \deff 
\begin{pmatrix}
    1 & 0 \\
    0 & \tilde{R} 
\end{pmatrix}.
$$
We introduce a generic $n$-tuple $(N_1,\dots,N_n)$
of Lipschitz continuous vector fields along 
$\Gamma \upharpoonright (s_0-\eps,s_0+\eps)$
by setting
\begin{equation}\label{rotation}
  \begin{pmatrix}
    N_1 \\
    \vdots \\
    N_n
  \end{pmatrix}
  \deff
  \tilde{R}
  \begin{pmatrix}
    M_1 \\
    \vdots \\
    M_n
  \end{pmatrix}.
\end{equation}
Because of the orthogonality relation $\tilde{R}^{-1} = \tilde{R}^T$,
the vector fields $N_1,\dots,N_n$ are orthonormal.
Their derivatives satisfy the equations
$$
  \begin{pmatrix}
    N_1 \\
    \vdots \\
    N_n
  \end{pmatrix}'
  = B
  \begin{pmatrix}
    N_1 \\
    \vdots \\
    N_n
  \end{pmatrix}
$$
with 
$$
  B \deff R' R^{-1} + R A R^{-1} =
  \begin{pmatrix}
    0 & a^T \tilde{R}^T \\
    -\tilde{R} a & (\tilde{R}'+\tilde{R}\tilde{A}) \tilde{R}^T
  \end{pmatrix} .
$$
Comparing this matrix with the matrix appearing in~\eqref{frame.bis},
it follows that $(T,N_1,\dots,N_n)$ will be the desired relatively parallel 
adapted frame of $\Gamma \upharpoonright (s_0-\eps,s_0+\eps)$ 
(with $a^T \tilde{R}^T = (k_1,\dots,k_n)$)
provided that $\tilde{R}$ is a solution of the initial value problem
\begin{equation}\label{initial}
\left\{
\begin{aligned}
  \tilde{R}'+\tilde{R}\tilde{A} &= 0 
  & \mbox{on}& \quad (s_0-\eps,s_0+\eps) \,, \\
  \tilde{R} &= \tilde{R}_0
  & \mbox{at}& \quad s_0 \,,
\end{aligned}
\right.
\end{equation}
where $\tilde{R}_0$ is an orthogonal matrix such that 
$$
  \begin{pmatrix}
    N_1^0 \\
    \vdots \\
    N_n^0
  \end{pmatrix}
  =
  \tilde{R}_0
  \begin{pmatrix}
    M_1(s_0) \\
    \vdots \\
    M_n(s_0)
  \end{pmatrix}.
$$
By standard results (see, \eg, \cite[Thm.~1.2.1]{Zettl}),
it follows that~\eqref{initial} has a unique 
absolutely continuous solution~$\tilde{R}$.
From the differential equation in~\eqref{initial},
we deduce that~$\tilde{R}$ is actually Lipschitz continuous 
under our hypotheses and that it is orthogonal.

\smallskip
\noindent
\emph{Global existence of the relatively parallel adapted frame.}
Let $J$ be any open precompact subinterval of~$I$ containing the point~$s_0$.
Since the curvature is bounded in~$J$, the interval can be covered 
by a finite number of open intervals of equal length (\cf~\eqref{epsilon}),
for each of which there exists a family of relatively parallel
adapted frames by the local construction above. 
To get the global relatively parallel adapted frame on~$J$
satisfying the desired initial condition at~$s_0$,
we can patch together the local ones by employing 
the local frame already constructed on $(s_0-\eps,s_0+\eps)$
and the freedom of choosing the initial condition
in the problem analogous to~\eqref{initial}
for the covering subintervals.
Smoothness at the point where they link together
is a consequence of the uniqueness part.
Since there is the desired relatively parallel adapted frame
on \emph{any} open precompact subinterval~$J$ of~$I$, the result follows.
\end{proof}
\begin{remark}
Let $\Gamma:I\to\R^3$ be a space curve  
for which the Frenet frame $(T,M_1,M_2)$ exists,
see Example~\ref{Ex.Frame}.
Let $(T,N_1,N_2)$ denote a relatively parallel adapted frame of~$\Gamma$.
Let us parameterise the rotation matrix~$\tilde{R}$ from~\eqref{rotation}
as follows: 
$$
  \begin{pmatrix}
    N_1 \\ N_2
  \end{pmatrix}
  =
  \begin{pmatrix}
    \cos\vartheta & -\sin\vartheta \\
    \sin\vartheta & \cos\vartheta
  \end{pmatrix}
  \begin{pmatrix}
    M_1 \\ M_2
  \end{pmatrix}
  ,
$$ 
where $\vartheta:I\to\R$ is a differentiable function.
It follows from~\eqref{initial} that $\vartheta'=\tau$.
That is, the normal vectors of any relatively parallel adapted frame of~$\Gamma$
are rotated with respect to the Frenet frame
with the angle being a primitive of the torsion.
\end{remark}
%

%----------------------------%
\subsection*{Acknowledgement}
%----------------------------%
%
The research of D.K.\ was partially supported by the GACR grant No.\ 18-08835S.

%\newpage
%------------------%
% BIBLIOGRAPHY
%------------------%
%
\bibliographystyle{amsplain}
\bibliography{bib}

\providecommand{\bysame}{\leavevmode\hbox to3em{\hrulefill}\thinspace}
\providecommand{\MR}{\relax\ifhmode\unskip\space\fi MR }
\providecommand{\MRhref}[2]{%
  \href{http://www.ams.org/mathscinet-getitem?mr=#1}{#2}
}
\providecommand{\href}[2]{#2}
\begin{thebibliography}{10}

\bibitem{Barseghyan-Khrabustovskyi_2018}
D.~Barseghyan and A.~Khrabustovskyi, \emph{Spectral estimates for {D}irichlet
  {L}aplacian on tubes with exploding twisting velocity}, Oper. Matrices
  \textbf{13} (2019), 311--322.

\bibitem{Bishop_1975}
R.~L. Bishop, \emph{There is more than one way to frame a curve}, The American
  Mathematical Monthly \textbf{82} (1975), 246--251.

\bibitem{Bruneau-Miranda-Parra-Popoff}
V.~Bruneau, P.~Miranda, D.~Parra, and N.~Popoff, \emph{Eigenvalue and resonance
  asymptotics in perturbed periodically twisted tubes: Twisting versus
  bending}, arXiv:1903.10599 (2019).

\bibitem{Bruneau-Miranda-Popoff_2018}
V.~Bruneau, P.~Miranda, and N.~Popoff, \emph{Resonances near thresholds in
  slightly twisted waveguides}, Proc. Amer. Math. Soc. \textbf{146} (2018),
  4801--4812.

\bibitem{ChDFK}
B.~Chenaud, P.~Duclos, P.~Freitas, and D.~Krej\v{c}i\v{r}\'{\i}k,
  \emph{Geometrically induced discrete spectrum in curved tubes}, Differential
  Geom. Appl. \textbf{23} (2005), no.~2, 95--105.

\bibitem{deOliveira-Hartmann-Verri_2019}
C.~R. de~Oliveira, L.~Hartmann, and A.~A. Verri, \emph{Effective {H}amiltonians
  in surfaces of thin quantum waveguides}, J. Math. Phys. \textbf{60} (2019),
  022101.

\bibitem{Oliveira-Rossini}
C.~R. de~Oliveira and A.~F. Rossini, \emph{Effective operators for {R}obin
  {L}aplacian in thin two- and three-dimensional curved waveguides}, preprint.

\bibitem{Oliveira-Verri_2017b}
C.~R. de~Oliveira and A.~A. Verri, \emph{Mild singular potentials as effective
  {L}aplacians in narrow strips}, Math. Scand. \textbf{120} (2017), 145--160.

\bibitem{Oliveira-Verri_2017a}
\bysame, \emph{Norm resolvent approximation of thin homogeneous tubes by
  heterogeneous ones}, Commun. Contemp. Math. \textbf{19} (2017), 1650060.

\bibitem{EKK}
T.~Ekholm, H.~Kova{\v{r}}{\'\i}k, and D.~Krej\v{c}i\v{r}\'{\i}k, \emph{A
  {H}ardy inequality in twisted waveguides}, Arch. Ration. Mech. Anal.
  \textbf{188} (2008), 245--264.

\bibitem{ES}
P.~Exner and P.~{\v S}eba, \emph{Bound states in curved quantum waveguides},
  J.~Math.~Phys. \textbf{30} (1989), 2574--2580.

\bibitem{Ferreira-Mascarenhas-Piatnitski_2015}
R.~Ferreira, L.~M. Mascarenhas, and A.~Piatnitski, \emph{Spectral analysis in
  thin tubes with axial heterogeneities}, Portugal. Math. \textbf{72} (2015),
  247--266.

\bibitem{FrHe}
R.~Froese and I.~Herbst, \emph{Realizing holonomic constraints in classical and
  quantum mechanics}, Commun.~Math.~Phys. \textbf{220} (2001), 489--535.

\bibitem{Gilbarg-Trudinger}
D.~Gilbarg and N.~S. Trudinger, \emph{Elliptic partial differential equations
  of second order}, Springer-Verlag, Berlin, 1983.

\bibitem{Jimbo-Kurata_2016}
S.~Jimbo and K.~Kurata, \emph{Asymptotic behavior of eigenvalues of the
  {L}aplacian on a thin domain under the mixed boundary condition}, Indiana
  Univ. Math. J. \textbf{65} (2016), 867--898.

\bibitem{Kato}
T.~Kato, \emph{Perturbation theory for linear operators}, Springer-Verlag,
  Berlin, 1966.

\bibitem{Keller-Teufel_2016}
J.~von Keller and S.~Teufel, \emph{The {NLS} limit for bosons in a quantum
  waveguide}, Ann. H. Poincar{\'e} \textbf{17} (2016), 3321--3360.

\bibitem{Kli}
W.~Klingenberg, \emph{A course in differential geometry}, Springer-Verlag, New
  York, 1978.

\bibitem{KK3}
M.~Kolb and D.~Krej\v{c}i\v{r}\'{\i}k, \emph{The {B}rownian traveller on
  manifolds}, J. Spectr. Theory \textbf{4} (2014), 235--281.

\bibitem{K1}
D.~Krej\v{c}i\v{r}\'{\i}k, \emph{Quantum strips on surfaces}, J.~Geom. Phys.
  \textbf{45} (2003), no.~1--2, 203--217.

\bibitem{K3}
\bysame, \emph{Hardy inequalities in strips on ruled surfaces}, J. Inequal.
  Appl. \textbf{2006} (2006), Article ID 46409, 10 pages.

\bibitem{K6-with-erratum}
D.~Krej\v{c}i\v{r}\'{\i}k, \emph{Twisting versus bending in quantum
  waveguides}, Analysis on Graphs and its Applications, Cambridge, 2007
  (P.~Exner et~al., ed.), Proc. Sympos. Pure Math., vol.~77, Amer. Math. Soc.,
  Providence, RI, 2008, pp.~617--636. See arXiv:0712.3371v2 [math--ph] (2009)
  for a corrected version.

\bibitem{K11}
\bysame, \emph{Waveguides with asymptotically diverging twisting}, Appl. Math.
  Lett. \textbf{46} (2015), 7--10.

\bibitem{KTdA2}
D.~Krej\v{c}i\v{r}\'{\i}k and R.~Tiedra de~Aldecoa, \emph{Ruled strips with
  asymptotically diverging twisting}, Ann. H. Poincar\'e \textbf{19} (2018),
  2069--2086.

\bibitem{KKriz}
D~Krej\v{c}i\v{r}\'{\i}k and J.~K\v{r}\'{\i}\v{z}, \emph{On the spectrum of
  curved quantum waveguides}, Publ.~RIMS, Kyoto University \textbf{41} (2005),
  no.~3, 757--791.

\bibitem{KL}
D.~Krej\v{c}i\v{r}\'{\i}k and Z.~Lu, \emph{Location of the essential spectrum
  in curved quantum layers}, J. Math. Phys. \textbf{55} (2014), 083520.

\bibitem{KRRS2}
D.~Krej\v{c}i\v{r}\'{\i}k, N.~Raymond, J.~Royer, and P.~Siegl, \emph{Reduction
  of dimension as a consequence of norm-resolvent convergence and
  applications}, Mathematika \textbf{64} (2018), 406--429.

\bibitem{KSed}
D.~Krej\v{c}i\v{r}\'{\i}k and H.~\v{S}ediv\'akov\'a, \emph{The effective
  {H}amiltonian in curved quantum waveguides under mild regularity
  assumptions}, Rev. Math. Phys. \textbf{24} (2012), 1250018.

\bibitem{Lampart-Teufel_2017}
J.~Lampart and S.~Teufel, \emph{The adiabatic limit of {S}chr{\"o}dinger
  operators on fibre bundles}, Math. Anal. \textbf{367} (2017), 1647--1683.

\bibitem{Mamani-Verri_2018b}
C.~R. Mamani and A.~A. Verri, \emph{Absolute continuity and band gaps of the
  spectrum of the {D}irichlet {L}aplacian in periodic waveguides}, Bull. Braz.
  Math. Soc. \textbf{49} (2018), 495--513.

\bibitem{Mamani-Verri_2018a}
\bysame, \emph{Influence of bounded states in the {N}eumann {L}aplacian in a
  thin waveguide}, Rocky Mt. J. Math. \textbf{48} (2018), 1993--2021.

\bibitem{Mehats-Raymond_2017}
F.~M{\'e}hats and N.~Raymond, \emph{Strong confinement limit for the nonlinear
  {S}chr{\"o}dinger equation constrained on a curve}, Ann. H. Poincar{\'e}
  \textbf{18} (2017), 281--306.

\bibitem{Spivak1}
M.~Spivak, \emph{A comprehensive introduction to differential geometry},
  vol.~I, Publish or Perish, Houston, Texas, 2005.

\bibitem{Verri_2019}
A.~A. Verri, \emph{{D}irichlet {L}aplacian in a thin twisted strip}, Int. J.
  Math. \textbf{30} (2019), 1950006.

\bibitem{Yachimura_2018}
T.~Yachimura, \emph{Two-phase eigenvalue problem on thin domains with {N}eumann
  boundary condition}, Differ. Integral. Equ. \textbf{31} (2018), 735--760.

\bibitem{Zettl}
A.~Zettl, \emph{Sturm-{L}iouville theory}, Amer. Math. Soc., 2010.

\end{thebibliography}
\end{document}